%% file: main.tex
\documentclass[nonconf,sigconf]{acmart}
\renewcommand\footnotetextcopyrightpermission[1]{} 

\AtBeginDocument{%
  \providecommand\BibTeX{{%
    \normalfont B\kern-0.5em{\scshape i\kern-0.25em b}\kern-0.8em\TeX}}}

\copyrightyear{2021}
\acmYear{2021}
\setcopyright{acmcopyright}
\acmConference[HSCC '21]{24th ACM InternationalConference on Hybrid Systems: Computation and Control}{May 19--21,2021}{Nashville, TN, USA}
\acmBooktitle{24th ACM International Conference on Hybrid Systems: Computationand Control (HSCC '21), May 19--21, 2021, Nashville, TN, USA}
\acmPrice{15.00}
\acmDOI{10.1145/3447928.3456649}
\acmISBN{978-1-4503-8339-4/21/05}

\acmSubmissionID{art_53}



\usepackage{xspace}
\usepackage{subcaption}
\usepackage{todonotes}
\usepackage{mathtools}

\input{commands}

\newif\ifarxiv
\arxivtrue
\ifarxiv
\usepackage{fancyhdr}
\fancypagestyle{arxiv}{%
  \fancyhf{} 
  
  \fancyhead[C]{\textbf{Updated from version in the 2021 Proc. of the ACM Int. Conf. on Hybrid Systems: Computation and Control (HSCC 2021)}}
}%
\fi

\ifarxiv
\else\fi
\begin{document}
\title{Strategy Synthesis for Partially-known Switched Stochastic Systems}
\author{John Jackson}
\email{john.m.jackson@colorado.edu}
\affiliation{%
  \institution{University of Colorado Boulder}
  \city{Boulder}
  \state{CO}
  \country{USA}
}
\author{Luca Laurenti}
\email{l.laurenti@tudelft.nl}
\affiliation{%
  \institution{Delft University of Technology}
  \city{Delft}
  \country{Netherlands}
}
\author{Eric Frew}
\email{eric.frew@colorado.edu}
\affiliation{%
  \institution{University of Colorado Boulder}
  \city{Boulder}
  \state{CO}
  \country{USA}
}
\author{Morteza Lahijanian}
\email{morteza.lahijanian@colorado.edu}
\affiliation{%
  \institution{University of Colorado Boulder}
  \city{Boulder}
  \state{CO}
  \country{USA}
}
\begin{abstract}

We present a data-driven framework for strategy synthesis for partially-known switched stochastic systems.
The properties of the system
are specified using \textit{linear temporal logic (\LTL) over finite traces} (\LTLf), which is as expressive as \LTL and enables interpretations over finite behaviors.
The framework first learns the unknown dynamics via Gaussian process regression.  Then, it builds a formal abstraction of the switched system in terms of an uncertain Markov model, namely an \emph{Interval Markov Decision Process (\imdp)}, by accounting for both the stochastic behavior of the system and the uncertainty in the learning step.
Then, we synthesize a strategy on the resulting \imdp that maximizes the satisfaction probability of the \LTLf specification and is robust against all the uncertainties in the abstraction.  This strategy is then refined into a switching strategy for the original stochastic system.
We show that this strategy is near-optimal and provide a bound on its distance (error) to the optimal strategy.
We experimentally validate our framework on various case studies, including both linear and non-linear switched stochastic systems.
\end{abstract}

\maketitle
\thispagestyle{arxiv}
\pagestyle{arxiv}
\input{sections/1_intro}

\input{sections/2_problem}

\input{sections/3_prelim}

\input{sections/4_method}
\input{sections/5_examples}
\input{sections/6_conclusion}

\section{Acknowledgement}
This work was supported in part by
the NSF Center for Unmanned Aircraft Systems, NSF under
awards IIP - 1650468, and NSF 2039062.

\bibliographystyle{ACM-Reference-Format}
\bibliography{refs}

\end{document}

%% file: commands.tex
\newcommand{\reals}{\mathbb{R}}
\newcommand{\naturals}{\mathbb{N}}

\newcommand\indicator{\textnormal{\textbf{1}}}
\newcommand{\interindicator}[2]{\indicator^{#2}_{#1}}

\newcommand{\x}{\textbf{{x}}}   
\newcommand{\pu}{\textbf{u}}

\newcommand{\policy}{\pi}
\newcommand{\optpolicy}{\policy^*}
\newcommand{\nearoptpolicy}{\policy^{\opterror}}
\newcommand{\Setofpolicies}{\Pi}
\newcommand{\FullSet}{X}

\newcommand{\distBound}{\epsilon}
\newcommand{\noise}{\textbf{v}}
\newcommand{\subgparam}{\theta}
\newcommand{\dataset}{\mathrm{D}}
\newcommand{\distribution}{p}
\newcommand{\gest}{\hat g}
\renewcommand{\path}{\omega}
\newcommand{\pathX}{\path_{\x}}

\newcommand{\PathFin}{\Omega_{\x}^{\text{fin}}}
\newcommand{\kernel}{k}
\newcommand{\ROISet}{R}
\newcommand{\roi}{r}
\newcommand{\opterror}{\varepsilon}

\newcommand{\SHS}{SHS\xspace}

\newcommand{\Spec}{\varphi}
\newcommand{\LTLf}{LTL$_f$\xspace}
\newcommand{\LTL}{LTL\xspace}
\newcommand{\Globally}{\mathcal{G}}
\newcommand{\Eventually}{\mathcal{F}}
\newcommand{\Until}{\mathcal{U}}
\newcommand{\Next}{\mathcal{X}}
\newcommand{\trace}{\rho}
\newcommand{\APs}{AP}
\newcommand{\Prop}{\mathfrak{p}}
\newcommand{\unsafe}{u}
\newcommand{\safe}{s}

\newtheorem{assumption}{Assumption}
\newtheorem{definition}{Definition}
\newtheorem{problem}{Problem}
\newtheorem{theorem}{Theorem}
\newtheorem{lemma}{Lemma}

\newtheorem{corollary}{Corollary}

\renewcommand{\kernel}{\kappa}
\newcommand{\mean}{\mu}
\newcommand\InfoGainBound{\gamma^m_\kernel}

\newcommand{\dfa}{\textsc{DFA}\xspace}
\newcommand{\DFAA}{\mathcal{A}_\Spec}
\newcommand{\DFAStates}{S}

\newcommand{\DFATransition}{\delta}
\newcommand{\DFAState}{s}
\newcommand{\DFAFinalState}{S_F}

\newcommand{\imdp}{\textsc{IMDP}\xspace}

\newcommand{\pimdp}{\textsc{PIMDP}\xspace}
\newcommand{\PIMDPS}{\mathcal{P}}

\newcommand{\I}{\mathcal{I}}

\newcommand{\probDist}{\mathcal{D}}

\newcommand{\feasibleDist}[2]{\gamma_{#1}^{#2}}

\newcommand{\Amdp}{A}

\newcommand{\amdp}{a}
\newcommand{\qmdp}{q}

\newcommand{\Qimdp}{Q}
\newcommand{\Aimdp}{A}
\newcommand{\qimdp}{q}

\newcommand{\qimdpprime}{\qimdp^\prime}

\newcommand{\qunsafe}{\qimdp_{\unsafe}}
\newcommand{\aimdp}{u}

\newcommand{\Pup}{\hat{P}}
\newcommand{\Plow}{\check{P}}
\newcommand{\plow}{\check{\p}}
\newcommand{\p}{p}
\newcommand{\pup}{\hat{\p}}
\newcommand{\pthresh}{\bar{\p}}
\newcommand{\popt}{\p^*}
\newcommand{\Qyes}{\Qimdp^{\text{yes}}}
\newcommand{\Qno}{\Qimdp^{\text{no}}}
\newcommand{\Qposs}{\Qimdp^{\text{?}}}

\newcommand{\pathimdp}{\omega_{\I}}

\newcommand{\pathimdpfin}{\pathimdp^{\mathrm{fin}}}

\newcommand{\Pathimdp}{\mathit{Paths}}
\newcommand{\Pathimdpfin}{\mathit{Paths}^{\mathrm{fin}}}
\newcommand{\last}{\mathit{last}}
\newcommand{\adversary}{\xi}
\newcommand{\Adversary}{\Xi}
\newcommand{\set}[1]{\{#1\}}

%% file: sections/1_intro.tex
\section{Introduction}
\label{sec:intro}

Switched stochastic systems are a class of \textit{stochastic hybrid systems} ({\SHS}s) that provide a powerful framework for modeling complex real-world systems.  They consist of a finite set of stochastic processes that capture the uncertainty in the evolution of the underlying system with the ability to switch between these processes, representing control options.  These models are employed in numerous application domains such as robotics \cite{luna:wafr:2014}, biological systems \cite{julius2008stochastic}, and cyber-physical systems \cite{haesaert2017certified}.  Many of the applications are in \textit{safety-critical} domains and require formal analysis of the underlying system.  
Existing formal approaches to analysis and synthesis of {\SHS}s are model based,
and the resulting guarantees apply only to the model of the system.
In reality, the true model of the system is often partially-known due to, e.g., the use of black-box controllers, or the lack of a closed-form analytical representation.  
This poses a major challenge for formal reasoning, which also relates to the classical question of \textit{how to extend formal guarantees from models to systems?}  
This work investigates a data-driven approach to address this challenge.

Formal verification and synthesis for {\SHS}s has been well studied in recent years, e.g., \cite{doyen2018verification,kushner2013numerical,soudjani2015fau,lahijanian2015formal,laurenti2020formal,cauchi2019efficiency}.  The proposed approaches can be generally divided into two categories.  One is a set of approaches based on numerical analysis of stochastic differential (difference) equations  with  asymptotic guarantees in terms of weak convergence \cite{kushner2013numerical}.  The other set of approaches is based on a finite \textit{abstraction} of the \SHS to a Markov process, and their formal guarantees are on probabilistic satisfaction of temporal logic specifications, namely \textit{linear temporal logic} (\LTL) and \textit{probabilistic computation tree logic} (PCTL) \cite{baier2008principles}.  Despite the recent advances, both categories of approaches assume that the \SHS model is fully known and perfectly represents the underlying system.  This assumption, however, is often violated, 
especially in modern systems where AI-modules are increasingly employed as black-box components.

A few recent studies focus on dealing with unknown dynamical systems, e.g., \cite{Dutta:IFAC:2018,Haesaert:Automatica:2017,Kenanian:Automatica:2019,Ahmadi:CDC:2017}.  The proposed approaches are based on data-driven methods and assume some knowledge on the system.  Work  \cite{Haesaert:Automatica:2017,Kenanian:Automatica:2019} impose a strong assumption that the underlying system is linear.  Then, they employ techniques such as Bayesian inference and chance-constrained optimization to provide probabilistic guarantees for the unknown system from a finite set of data.  
Work \cite{Ahmadi:CDC:2017} relaxes the linearity assumption and proposes approximation of the unknown dynamics through a piecewise-polynomial function.  Then, the safety of the system is assessed through barrier certificates.  While this method can deal with a more general class of systems, it is unclear how the guarantees can be extended to the underlying system.

An effective method to deal with unknown dynamics in safety-critical applications is \textit{Gaussian process} (GP) regression \cite{rasmussen2003gaussian}. 
The advantage of GP regression is in its ability to quantify the bound on the uncertainty in the learning process as derived in \cite{srinivas2012information,Germain2016pac,chowdhury2017kernelized, lederer2019uniform}. 
This has led to an increased use of GPs in safe learning frameworks, e.g., \cite{akametalu2014reachability,sui2015safe,berkenkamp2017safe,polymenakos2019safety,jackson2020safety}.  
In most work, the main objective is to learn safe policies via reinforcement learning with the exception of work \cite{jackson2020safety}, which considers a safety verification problem of unknown systems with noisy measurements.  The proposed framework uses GPs to construct a Markov abstraction for an invariant set (safety) analysis from a noisy dataset.  In all these work, the assumption is that the underlying system is deterministic and the specification is simple, whereas the focus here is on unknown stochastic systems with complex specifications.

This work presents a formal synthesis framework for stochastic systems with partially-known models in the form of switched stochastic processes.  
The framework is able to provide formal guarantees on the behavior of the underlying system from a set of data.  
The specification language is \textit{\LTL over finite traces} (\LTLf) \cite{de2013linear}, which has the same expressively as \LTL, but the interpretations of its formulas are over finite behaviors making it an appropriate language for highly uncertain (unknown) systems such as those considered here.  
The approach is based on finite abstraction and employs GP regression for its construction.
Given a set of data, the framework first learns the unknown dynamics using GP regression.  
Then, an abstraction is constructed in the form of an uncertain Markov process, namely \textit{interval Markov decision process} (\imdp) using the known and learned dynamics as well as the errors bounds of the learning process.  
Given an \LTLf property, a strategy is computed on the abstraction that maximizes the probability of satisfaction of the property and is robust against all the errors introduced in the learning and abstraction steps.  
This not only results in a switching (control) strategy for the underlying system, but it also provides a lower bound probability for the satisfaction of the \LTLf property for every initial state.

The main contribution of this work is a theoretical and computational framework for control synthesis for partially-known stochastic systems from a given set of data.  
This work shows a method of harvesting the power of machine learning techniques, in particular GP regression, in a formal synthesis framework.
Unlike classical model-based approaches, this framework enables the extension of the formal guarantees to the underlying system.  
This is achieved by formally incorporating both the uncertainty related to the stochastic behavior of the system and the uncertainty related to the partial knowledge of the system in the abstraction, and then accounting for these uncertainties in generating a robust switching strategy.
As a result, this framework allows for synthesis for complex systems from simplified (low-fidelity) models, i.e., linearized models; hence, enabling the use of rich and matured techniques for simple (linear) models in control design for complex systems.  
Furthermore, this paper presents derivations for probabilistic bounds for the transition probabilities of the \imdp abstraction as well as proofs of correctness for the methodology.
Finally, the synthesis framework is demonstrated through a series of case studies on unknown stochastic systems with both linear and nonlinear dynamics with various \LTLf specifications.

%% file: sections/2_problem.tex
\section{Problem Formulation}
\label{sec:problem}

Consider a partially-known switched stochastic process as described below:
\begin{equation}\label{eq:system}
    \x_{k+1} = f_{\pu_k}(\x_k) + g_{\pu_k}(\x_k) + \noise_k, 
\end{equation}
where $k\in \naturals,$ $\x_k \in \reals^n$,  $\pu_k \in U$, and $U=\{1,...,m \}$ is a finite set of \textit{modes} or \textit{actions}. For every $u\in U,$ $f_u:\mathbb{R}^n \to \mathbb{R}^n$ is a (known a-priori) continuous function and $g_u:\mathbb{R}^n\to \mathbb{R}^n$ is a possibly nonlinear continuous function representing the unknown dynamics of Process \eqref{eq:system}.
The noise term $\noise_k$ is a random variable with an independent and stationary $\subgparam$-sub-Gaussian distribution $\distribution_{v}$. 
This class of distributions are those whose tails decay at least as fast as a Gaussian with variance $\subgparam^2$, including all distributions with bounded support the Gaussian distribution itself \cite{massart2007concentration}.

Intuitively, $\x_k$ is a stochastic process driven by the noise process $\noise_k$, where some or all the dynamics are unknown in each mode, and $\pu_k$ indicates the current mode (and hence switching between the modes).  
Process \eqref{eq:system} is a rich model that allows for the inclusion of modeling errors in addition to noise.  For instance, consider a nonlinear noisy control system with a finite set of controls $U$. If only a linear approximate model of the system is available, then Process~\eqref{eq:system} can be used to represent it, where $f_u$ becomes the approximate linear model of the system and $g_u$ is all the higher-order dynamics that are not modelled under each controller $u$.

We assume to have a collection of state-action-state measurements $\dataset=\{(x_i,u_i,x^+_i)_{i=1}^m\}$ generated by Process~\eqref{eq:system}, where $x^+_i\in\reals^n$ is a sample of one-step evolution of Process~\eqref{eq:system} when it is initialized at $x_i\in\reals^n$ in mode $u_i\in U$.  Our goal is to use $\dataset$ to learn $g_u$ for each $u\in U$. In order to achieve this correctly, we need an assumption on the regularity of $g_u$. The following assumption suffices to guarantee that $g_u$ can be learned arbitrarily well via GP regression.
\begin{assumption}
    \label{assump:smoothnes}
    For a compact set $\FullSet\subset \mathbb{R}^n$, let $\kernel:\mathbb{R}^n\times \mathbb{R}^n\to \mathbb{R}_{> 0}$ be a given kernel and $\mathcal{H}_\kernel(\FullSet)$ the reproducing kernel Hilbert space (RKHS) of functions over $\FullSet$ corresponding to $\kernel$ with norm $\| \cdot \|_\kernel$ \cite{srinivas2012information}.  Then, for each $u \in U$  and $i\in \{1,...,n \},$ $g^{(i)}_u(\cdot) \in \mathcal{H}_\kernel(\FullSet)$ and for a constant $B_i>0,$ it holds that  $\| g^{(i)}_u(\cdot) \|_\kernel \leq B_i$, where $g_u^{(i)}$ is the $i$-th component of $g_u$.
\end{assumption}
\noindent
Assumption \ref{assump:smoothnes} is a standard assumption \cite{srinivas2012information,jackson2020safety}, which is intimately related to the continuity of $g_u$, as 
discussed
in Section \ref{sec:IMDP}. 
For instance, assuming that $\kernel$ is the widely used squared exponential kernel, we obtain that $\mathcal{H}_\kernel(\FullSet)$ is a space of functions that is dense with respect to the set of continuous functions on a compact set $\FullSet \subset \reals^n,$ i.e., members of $\mathcal{H}_\kernel(\FullSet)$ can approximate any continuous function on $\FullSet$ arbitrarily well \cite{steinwart2001influence}.

Let $\pathX = x_0 \xrightarrow{u_0} x_1 \xrightarrow{u_1}  \ldots $ be a sample path (trajectory) of Process \eqref{eq:system}  and denote by $\pathX(k)=x_{k},$  the state of $\pathX$ at time $k$.  
Further, we denote by $\PathFin$ the set of all sample paths with finite length, i.e, the set of trajectories $\pathX^{k}= x_0 \xrightarrow{u_0} x_1 \xrightarrow{u_1}  \ldots \xrightarrow{u_{k-1}} x_k$ for all $k\in \naturals$. With a slight abuse of notation, given path $\pathX$, we denote by $\pathX^k$ the prefix of length $k+1$ of $\pathX$.  

Given a finite path, a \textit{switching strategy} chooses the mode (action) of Process~\eqref{eq:system}.

\begin{definition}[Switching Strategy]
    \label{def:switching-strategy}
    A switching strategy $\policy_\x: \PathFin \to U$ is a function that maps a finite path $\pathX^{k}\in \PathFin$ to a mode (action) $u \in U$. The set of all switching strategies is denoted by $\Setofpolicies_\x.$
\end{definition}
\noindent
For $u\in U$, 
a Borel measurable set $\FullSet\subseteq \reals^n$, and $x\in \reals^n,$ 
call 
$$T^{u}(\FullSet \mid x)=\int \indicator_\FullSet(f_u(x)+g_u(x)+v)\distribution_{v}(\bar v)d \bar v,$$ the stochastic transition function  induced by Process \eqref{eq:system} in mode $u\in U$, where 
\begin{equation*}
    \indicator_\FullSet(x) =
    \begin{cases}
        1 & \text{if } x\in \FullSet\\
        0 & \text{otherwise}
    \end{cases}
\end{equation*}
is the indicator function.
From the definition of $T^{u}(\FullSet\mid x)$  it follows that, given a strategy $\policy_\x$, for a time horizon $[0,N],$ Process \eqref{eq:system} defines a stochastic process on the canonical space $\Omega=(\mathbb{R}^n)^{N+1}$ with the Borel $\sigma-$algebra $\mathcal{B}(\Omega)$ induced by the product topology and with the unique probability measure $P$ generated by $T^{\policy_\x}$ and a (fixed) initial condition $x_0\in \mathbb{R}^n$ 
such that for $k\in \{1,...,N\}$
\begin{align*}
    &P[\pathX^N(0)\in \FullSet] = \indicator_{\FullSet}(x_0),\\
    &P[\pathX^N(k) \in \FullSet \mid \pathX^N(k-1)=x, \policy_\x] = T^{\policy_\x(\pathX^{k-1})}(\FullSet \mid x). 
\end{align*}
Furthermore, for $N=\infty$, $P$ is still uniquely defined by $T^u$ by the \emph{Ionescu-Tulcea extension theorem} \cite{abate2014effect}.

In this paper, we are interested in the properties of Process~\eqref{eq:system} in a compact set $\FullSet \subset \reals^n$.  Specifically, we analyze the behavior of Process~\eqref{eq:system} with respect to a finite set of closed regions of interest $\ROISet = \{\roi_1,\ldots,\roi_{|\ROISet|}\}$, where $\roi_i \subseteq \FullSet$.  To this end, we associate to each region $\roi_i$ the atomic proposition $\Prop_i$ such that $\Prop_i = \top$ (i.e., $\Prop_i$ is \textit{true}) if $x \in \roi_i$; otherwise $\Prop_i = \bot$ (i.e., $\Prop_i$ is \textit{false}).  
Let $\APs = \{\Prop_1, \ldots,\Prop_{|\ROISet|}\}$ denote the set of all atomic propositions and $L: \FullSet \rightarrow 2^{\APs}$ be the labeling function that assigns to state $x$ the set of atomic propositions that are true at $x$, i.e.,
$$\Prop_i \in L(x) \quad \Leftrightarrow \quad x \in \roi_i.$$
Then, we define the \textit{trace} or \textit{observation} of path 
$\pathX^{k}= x_0 \xrightarrow{u_0} x_1 \xrightarrow{u_1}  \ldots \xrightarrow{u_{k-1}} x_k$
to be
\begin{equation*}
	\trace = \trace_0 \trace_1 \ldots \trace_k,
\end{equation*}
where $\trace_i = L(x_i) \in 2^{\APs}$ for all $i \leq k$. With an abuse of notation we use $L(\pathX^k)$ to denote the trace of $\pathX^k$.

\subsection{Linear temporal logic on finite traces (\LTLf)}
\label{sec:ltlf}
In this work, we are interested in the temporal properties of Process~\eqref{eq:system} with respect to the regions of interest in $R$.  To express such properties, \textit{linear temporal logic} (\LTL) \cite{baier2008principles} is a popular choice of language given its rich expressivity and intuitive formalism.  
Here, we employ \textit{\LTL interpreted over finite traces} (\LTLf) \cite{de2013linear}, which has the same syntax as \LTL but its semantics is defined over finite traces.
 
\begin{definition}[\LTLf Syntax]
    An \LTLf formula $\Spec$ is built from a set of atomic propositions $\APs$ and is closed under the Boolean connectives as well as the ``next'' operator $\Next$ and the ``until'' operator $\Until$:
	\begin{equation*}
		\Spec ::= \top   \mid   \Prop   \mid   \neg \Spec   \mid   \Spec \wedge \Spec   \mid  \Next \Spec  \mid  \Spec \Until \Spec
	\end{equation*}
	where $\Prop \in \APs$, $\top$ is ``true'' or a tautology, and $\neg$ and $\wedge$ are the ``negation'' and ``and'' operators in Boolean logic, respectively.
\end{definition}

\noindent
The common temporal operators ``eventually'' ($\Eventually$) and ``globally'' ($\Globally$) are defined as: 
$$\Eventually \, \Spec = \top \, \Until \, \Spec \quad \text{and} \quad \Globally \, \Spec = \neg \Eventually \, \neg \Spec.$$
The semantics of \LTLf is defined as follows.

\begin{definition}[\LTLf Semantics]
	The semantics of an \LTLf formula $\Spec$ are defined over finite traces in $\APs^*$. 
	The set of all finite traces is $(2^{\APs})^*$. 
	Let $|\trace|$ denote the length of trace $\trace$ and $\trace_i$ be the $i$-th symbol of $\trace$. Further, $\trace, i \models \Spec$ is read as: ``the $i$-th step of trace $\trace$ is a model of $\Spec$.'' Then,
	\begin{itemize}
		\item $\trace , i \models \top,$
		\item $\trace , i \models \Prop$ \; \text{iff} \; \; $\Prop \in \trace_i,$
		\item $\trace , i \models \neg \Spec$ \; \text{iff} \; $\trace, i \not \models \Spec,$
		\item $\trace , i \models \Spec_1 \wedge \Spec_2$ \; \text{iff} \; $\trace, i \models \Spec_1$ and $\trace, i \models \Spec_2,$
		\item $\trace , i \models \Next \Spec$ \; \text{iff} \; $|\trace| > i+1$ and $\trace, i+1 \models \Spec,$
		\item $\trace , i \models \Spec_1 \Until \Spec_2$ \; \text{iff} \; $\exists j$ s.t. $ i \leq j < |\trace|$ and $\trace, j \models \Spec_2$ and $\forall k$, $i \leq k < j$, \; $\trace, k \models \Spec_1.$
	\end{itemize}
	Finite trace $\trace$ satisfies $\Spec$, denoted by $\trace \models \Spec$, if $\trace, 0 \models \Spec$.
\end{definition}
An \LTLf formula $\Spec$ defines a language $\mathcal{L}(\Spec)$ over the alphabet $2^{\APs}$. $\mathcal{L}(\Spec)$ is a regular language, more specifically, $$\mathcal{L(\Spec)} = \{\trace \in (2^{\APs})^* \mid \trace \models \Spec\}.$$ 

Given compact set $\FullSet \subset \reals^n$, its set of regions of interest $\ROISet$ and the corresponding set of atomic propositions $\APs$, and an \LTLf formula $\Spec$ defined over $\APs$, as in \cite{wells2020ltlf}, we say that path $\pathX$ of Process~\eqref{eq:system} satisfies $\Spec$ if there exists a prefix of $\pathX$ that is in the language of $\Spec$ and lies entirely in $\FullSet$, i.e., 
\begin{multline}
    \label{eq:satisfying-path}
    \pathX \models \Spec  \quad \Leftrightarrow \quad \exists k \in \naturals  \;\; s.t. \;\; L(\pathX^k) \in \mathcal{L}(\Spec) \text{ and } \\ \pathX^k(k') \in \FullSet \;\; \forall k' \leq k,
\end{multline}
where $L(\pathX^k) \in \big(2^{\APs}\big)^*$ is the trace (observation) of $\pathX^k$.

\subsection{Problem Formulation}
 The ideal goal of this work is, given an \LTLf formula $\Spec$, to synthesize a switching strategy $\optpolicy_\x$ such that under $\optpolicy_\x$ the probability of the paths of Process \eqref{eq:system} that satisfy $\Spec$ is maximized. Nevertheless, we should stress that, in general, the partial knowledge of Process \eqref{eq:system} and the limited amount of data available (not controllable a-priori) make it infeasible to find a switching strategy that maximizes such a probability. Hence, in Problem \ref{prob:Syntesis} we seek a near-optimal strategy such that, under this switching strategy, Process \eqref{eq:system} is guaranteed to satisfy $\Spec$ with a (high) probability greater than a given threshold with a quantifiable distance from the optimal probability. 
 \begin{problem}[Switching Strategy Synthesis]
 \label{prob:Syntesis}
    Given a partially-known switched stochastic system as defined in Process \eqref{eq:system}, a dataset $\dataset,$ a compact set $X$, an \LTLf property $\Spec$ defined over the regions of interest in $X$, and a probability threshold $\pthresh$, find a near-optimal switching strategy $\nearoptpolicy_\x$ that determines whether for every $x_0 \in \FullSet$
    \begin{equation*}
        P[	\pathX \models \Spec \mid \nearoptpolicy_\x, \pathX(0) = x_0] \geq \pthresh,
    \end{equation*}
    and quantify the corresponding error $\opterror_{x_0} \geq 0$ with respect to the optimal switching strategy, i.e., 
    \begin{equation*}
        | P[\pathX \models \Spec \mid \nearoptpolicy_\x, \pathX(0)=x_0] -\popt(x_0)| \leq \opterror_{x_0},
    \end{equation*}
where  $\popt(x_0)=\max_{\policy_\x \in \Setofpolicies_\x } P[	\pathX \models \Spec \mid \policy_\x, \pathX(0) = x_0]$.
 \end{problem}

\subsubsection*{Overview of the Approach}
In order to solve Problem \ref{prob:Syntesis} we rely on GP regression  and Assumption \ref{assump:smoothnes} to find a function $\gest_u$ such that with high probability, $|\gest_u(x) - g_u(x)|\leq \distBound_u$ for all $x \in \FullSet$ and a given $\distBound_u>0$. We then use $\gest_u$ to build an abstraction of Process \eqref{prob:Syntesis} in terms of a finite Markov model, where the stochastic nature of Process \eqref{eq:system}, the error in employing $\gest_u$ instead of $g_u$, and the error corresponding to the discretization of space are all formally modelled as uncertainty. We then synthesize an optimal strategy for the resulting Markov model that maximizes the probability that the paths of the Markov model satisfy $\Spec$ and is robust against the uncertainties. Finally, we derive upper and lower bounds on the probability that Process \eqref{eq:system} satisfies $\Spec$ under this strategy.

%% file: sections/3_prelim.tex
\section{preliminaries}

\subsection{Gaussian Process Regression}
\label{sec:PrelimGP}
{Gaussian Process} (GP) regression is a non-parametric Bayesian machine learning method \cite{rasmussen2003gaussian}. For an unknown function  $\mathrm{g}:\mathbb{R}^n\to \mathbb{R}$, the basic assumption of GP regression is that $\mathrm{g}$ is a sample from a GP with covariance $\kernel: \mathbb{R}^n \times \mathbb{R}^n \to \mathbb{R}_{>0} $. 
Consider a dataset of samples $\dataset=\{(\mathrm{x}_i,\mathrm{y}_i),i\in \{1,\dots,m \}\}$, where 
$\mathrm{y}_i$ is a sample of an observation of $\mathrm{g}(\mathrm{x}_i)$ with independent zero-mean noise ${v}$, which is assumed to be normally distributed with variance $\sigma^2$. 
Let $\mathrm{X}$ and $\mathrm{Y}$ be ordered vectors with all points in $\dataset$ such that $\mathrm{X}_i = \mathrm{x}_i$ and $\mathrm{Y}_i = \mathrm{y}_i$.  Further, call $K(\mathrm{X},\mathrm{X})$ the matrix with $K_{i,j}(\mathrm{X}_i,\mathrm{X}_j)=\kernel(\mathrm{x}_i,\mathrm{x}_j)$, $K(\mathrm{x},\mathrm{X})$ the vector such that $K_{i}(\mathrm{x},\mathrm{X})=\kernel(\mathrm{x},\mathrm{X}_i)$, and $K(\mathrm{X},\mathrm{x})$ defined accordingly. Then, the predictive distribution of $\mathrm{g}$ at a test point $\mathrm{x}$ is given by the conditional distribution of $\mathrm{g}$ 
given $\dataset$, which is Gaussian and with mean $\mu_\dataset$ and variance $\sigma_\dataset^2$ given by
\begin{align}
    \begin{split}\label{eq:post-mean}
      & \mu_{\dataset}(\mathrm{x}) = K(\mathrm{x},\mathrm{X}) \big( K(\mathrm{X},\mathrm{X})+ \sigma^2I_{m} \big)^{-1} Y   \\
    \end{split}\\
    \begin{split}\label{eq:post-kernel}
      & \sigma_{\dataset}^2(\mathrm{x}) = \kernel(\mathrm{x},\mathrm{x})- K(\mathrm{x},\mathrm{X})\big( K(\mathrm{X},\mathrm{X})+ \sigma^2 I_{m} \big)^{-1}K(\mathrm{X},\mathrm{x}),
     \end{split}
\end{align}
where $I_{m}$ is the identity matrix of size $m \times m$.

\subsection{Interval Markov Decision Processes}
We use a generalization of Markov decision processes to abstract the system. An \textit{interval Markov decision process} (\imdp), also called bounded-parameter Markov decision process, uses interval-valued transition probabilities \cite{givan2000bounded,Hahn:TOMACS:2019}.

\begin{definition}[{\imdp}] \label{def:imdp}
    An interval Markov decision process ({\imdp}) is a tuple $\I = (\Qimdp,\Aimdp,\Plow,\Pup,\APs,L)$, where
    \begin{itemize}
    	\setlength\itemsep{1mm}
    	\item $\Qimdp$ is a finite set of states,
    	\item $\Aimdp$ is a finite set of actions, and $\Aimdp(\qimdp)$ denotes the set of available actions at state $\qimdp \in \Qimdp$.
        \item $\Plow: \Qimdp \times \Aimdp \times \Qimdp \rightarrow [0,1]$ is a function, where $\Plow(\qimdp,a,\qimdpprime)$ defines the lower bound of the transition probability from state $\qimdp \in \Qimdp$ to state $\qimdpprime \in \Qimdp$ under action $a \in \Aimdp(\qimdp)$,
        \item $\Pup: \Qimdp \times \Aimdp \times \Qimdp \rightarrow [0,1]$ is a function, where $\Pup(\qimdp,a,\qimdpprime)$ defines the upper bound of the transition probability from state $\qimdp \in \Qimdp$ to state $\qimdpprime \in \Qimdp$ under action $a \in \Aimdp(\qimdp)$,
        \item $\APs$ is a finite set of atomic propositions,
        \item $L: \Qimdp \rightarrow 2^{\APs}$ is a labeling function that assigns to each state $\qimdp \in \Qimdp$ a subset of $\APs$.
    \end{itemize}
\end{definition}
For all $\qimdp,\qimdpprime \in \Qimdp$ and $a \in \Aimdp(\qimdp)$, it holds that $\Plow(\qimdp,a,\qimdpprime) \leq \Pup(\qimdp,a,\qimdpprime)$ and $
    \sum_{\qimdpprime \in \Qimdp} \Plow(\qimdp,a,\qimdpprime) \leq 1 \leq \sum_{\qimdpprime \in \Qimdp} \Pup(\qimdp,a,\qimdpprime)$.
 
A path of an \imdp is a sequence of states $\pathimdp = \qmdp_0 \xrightarrow{\amdp_0} \qmdp_1 \xrightarrow{\amdp_1} \qmdp_2 \xrightarrow{\amdp_2}  \ldots$ such that $\amdp_k \in \Amdp(\qmdp_k)$ and $\Pup(\qmdp_k, \amdp_k, \allowbreak{\qimdp_{k+1}}) > 0$ for all $k \in \naturals$. We denote 
the last state of a finite path $\pathimdpfin$ by $\last(\pathimdpfin)$ and the set of all finite and infinite paths by $\Pathimdpfin$ and $\Pathimdp$, respectively. The actions are chosen according to a strategy $\policy$ which is defined below. 
\begin{definition}[Strategy]
\label{def:strategy}
    A strategy $\policy$ of an \imdp model $\I$ is
    a function $\policy: \Pathimdpfin \rightarrow \Aimdp$ that maps a finite path $\pathimdpfin$ of $\I$ onto an action in $\Aimdp(\last(\Pathimdpfin))$. If a strategy depends only on $\last(\pathimdpfin)$, it is called a memoryless (stationary) strategy.  
    The set of all strategies is denoted by $\Setofpolicies$.
\end{definition}

Given an arbitrary strategy $\policy$, we are restricted to the set of interval Markov chains defined by the transition probability intervals induced by $\policy$. In order to reduce this to a Markov chain, we define the adversary function, which assigns a transition probability distribution at each state.
\begin{definition}[Adversary]
\label{def:adversary}
    For an \imdp $\I$, an adversary is a function $\adversary: \Pathimdpfin \times \Aimdp \rightarrow \probDist(\Qimdp)$ that, for each finite path $\pathimdpfin \in \Pathimdpfin$, state $\qimdp=\last(\pathimdpfin)$, and action $a \in \Aimdp(\last(\pathimdpfin))$, assigns a feasible distribution $\feasibleDist{\qimdp}{a}$ which satisfies
    \begin{equation*}
    \Plow(\qimdp,a,\qimdpprime) \leq \feasibleDist{\qimdp}{a}(\qimdpprime) \leq \Pup(\qimdp,a,\qimdpprime).
    \end{equation*}
    The set of all adversaries is denoted by $\Adversary$.
\end{definition}

Under a strategy and a valid adversary, the \imdp collapses to a Markov chain and induces a probability measure on its paths. 
We use this measure as an optimization objective for synthesizing a desirable strategy.

%% file: sections/4_method.tex
\section{IMDP Abstraction}
\label{sec:IMDP}
In order to solve Problem \ref{prob:Syntesis}, we start by abstracting Process \eqref{eq:system} into an \imdp $\I = (\Qimdp,\Aimdp,\Plow,\Pup,\APs,L)$. Below we describe how we obtain $\Qimdp,\Aimdp, \APs,$ and $L$. Then, in Section \ref{Sec:transitionBOunds} we consider upper and lower bounds for the transition probabilities.

\subsection{States and Actions}
The set of states $\Qimdp$ of $\I$ is obtained by discretizating the compact set $\FullSet$.  This discretization needs to respect the set of regions of interest $\ROISet=\{\roi_1,...,\roi_{|\ROISet|} \}$.
To achieve this, we first construct a set of non-overlapping regions of interests $\ROISet'$ from $\ROISet$ such that 
\begin{equation*}
    \cup_{\roi' \in \ROISet'} \, \roi' = \ROISet
    \quad \text{and} \quad \roi'_i \cap \roi'_j = \emptyset \;\;\; \forall \roi'_i,\roi'_j \in \ROISet' \text{ and } \roi'_i \neq \roi'_j.
\end{equation*}
Then, we partition each $\roi'_i$ into a set of cells (regions) that are non-overlapping.  Next, we partition the remainder of the compact set ($\FullSet \setminus \ROISet$) to a set of non-overlapping cells.
Let $\Qimdp_\safe =\{ \qimdp_1,...,\qimdp_{|\Qimdp_\safe|} \}$ denote the resulting set of all cells (include $\ROISet'$).  
Then, by construction, it holds that
\begin{equation*}
	\cup_{\qimdp \in \Qimdp_\safe} \qimdp = \FullSet, \quad \text{and} \quad \qimdp \cap \qimdp' = \emptyset \;\;\; \forall \qimdp,\qimdp' \in \Qimdp_\safe \text{ and } \qimdp \neq \qimdp'.
\end{equation*}
Each region is associated to a state of \imdp $\I$.  With an abuse of notation, $\qimdp$ denotes both the region, i.e., $\qimdp \subset \FullSet$, as well as its corresponding \imdp state, i.e, $\qimdp \in \Qimdp$.  From the context, the correct interpretation of $\qimdp$ should be clear.  Furthermore, let $\qunsafe$ denote the remainder of the state space, i.e., $\reals^n \setminus \FullSet$.  Then, the set of states of $\I$ is defined as
\begin{equation*}
	\Qimdp = \Qimdp_\safe \cup \set{\qimdp_\unsafe}.
\end{equation*}
The set of actions $\Aimdp$ of $\I$ is given by the set of modes $U$, i.e., $\Aimdp = U$, and all actions are available at each state of $\I$, i.e., $\Aimdp(\qimdp) = \Aimdp$ for all $\qimdp \in \Qimdp$.

The set of atomic propositions $\APs$ is the same as those defined over $X$. With an abuse of notation, we define the \imdp labeling function $L:\Qimdp\to 2^{\APs}$ with $L(q) = L(x)$ for any choice of $x \in \qimdp$. Note that, because the discretization respects the regions of interests, the labels of the points in a discrete cell are necessarily the same, i.e., $L(x) = L(x')$ for all $x,x' \in \qimdp$.

\subsection{Transition Probability Bounds}
\label{Sec:transitionBOunds}
In order to compute the transition probability bounds $\Plow$ and $\Pup$ for all $\qimdp,\qimdp' \in \Qimdp$ and $\aimdp \in \Aimdp = U,$ we need to derive the following bounds:
\begin{align}
	\label{eqn:CriteriaTransitionLowerBound}
        \Plow(\qimdp,\aimdp,\qimdp') & \leq \min_{x\in \qimdp}  T^{\aimdp}(\qimdp'\mid x),  \\
        \Pup(\qimdp,\aimdp,\qimdp') & \geq \max_{x\in \qimdp}  T^{\aimdp}(\qimdp'\mid x).
\end{align}
However, without any knowledge about $g_u$ in Process \eqref{eq:system}, the computation of such quantities is infeasible. In what follows we show how we can employ the data in $\dataset$ and GP regression to compute a function $\gest_u$ such that for any $x\in \FullSet,$ $g_u(x)$ and $\gest_u(x)$ are provably close.

\subsubsection{Regression Approach}
In our setting, data in $\dataset$ are samples $(x,u,x^+)$ of Process \eqref{eq:system} such that
$$ x^+=f_u(x)+g_u(x)+v,  $$
where both $x$ and $u$ are known and $v$ is a sample from the noise process $\noise$. From this we can obtain a measurement of $g_u$ by simply noticing that $f_{u}$ is known, i.e., we obtain a dataset composed by:
\begin{align}
    \label{Eqn:NewDataset}
    y^+=x^+-f_u(x)=g_u(x)+v,
\end{align}
where $y^+,x^+,x, u$ are all known. Note that, in our setting,  we make no assumptions on the fact that $g_u$ is a sample from a given GP. Furthermore, the noise $\noise_k$ is not necessarily Gaussian for any $k\in\mathbb{N}$. As a result, the assumptions for GP regression discussed in Section \ref{sec:PrelimGP} are not satisfied and we cannot directly use its prediction to make probabilistic statements over $g_u$. 
Nevertheless, thanks to Assumption~\ref{assump:smoothnes} we can rely on the properties of the RKHS space generated by $\kernel$ to bound the regression error even in our more agnostic setting.

In particular, for each $g_u:\reals^n\to\reals^n$, we use $n$ independent GPs to learn $g_{u}^{(i)}$, the $i$-th component of $g_u$. 
Then, for a given mode $u$, we consider $\gest_u^{(i)}=\mu_{\dataset},$ where $\mu_{\dataset}$ is the posterior mean of the GP as described in \eqref{eq:post-mean}.
We use the following Lemma from \cite{chowdhury2017kernelized} to characterize the error in employing $\gest_u$ instead of $g_u$.
\begin{lemma}[\hspace{1sp}\cite{chowdhury2017kernelized}, Theorem 2]\label{th:RKHS}
    Let $\FullSet$ be a compact set, $\delta\in(0,1)$, $\InfoGainBound$ the maximum information gain parameter associated with $\kernel$ and dataset $\dataset$ of $m$ training points, and $B_i > 0$ such that $\|g_u^{(i)} \|_{\kernel}\leq B_i$. Assume that $\mathrm{\noise}$ is $\subgparam$-sub-Gaussian and $\mean_\dataset$ and $\sigma_\dataset$ are found by setting $\sigma=1+2/m$. Define $\beta=(\subgparam/\sqrt{\sigma}) \big( B_i + \subgparam\sqrt{2(\InfoGainBound + 1 + \log{1/\delta})} \big)$. Then, it holds that 
    \begin{equation}\label{eq:proberror}
        P\big[\forall x \in \FullSet, \; |\mean_D(x) - g_u^{(i)}(x)| < \beta\sigma_D(x) \big]\geq 1-\delta.
    \end{equation}
\end{lemma}
One challenge in employing Lemma~\ref{th:RKHS} is in determining the values (or bounds) for the information gain constant $\InfoGainBound$ and the RKHS constant $B_i$. A procedure for obtaining  $\InfoGainBound$ is given in~\cite{srinivas2012information}. The RKHS constant $B_i$ is instead intimately related to the continuity of $g_u$, as shown in Theorem 3.11 of~\cite{paulsen2016introduction}, where a bound of $B_i$ in terms of the maximum value that $g_u$ obtains in $\FullSet$ and the kernel $\kernel$ is given.

\subsubsection{Transitions within $\Qimdp_{\safe}$}
For all states $\qimdp,\qimdp' \in \Qimdp_\safe$, the transition probability bounds in  \eqref{eqn:CriteriaTransitionLowerBound}  are given by Theorem \ref{thm:transition} below. In order to state this result, we first need to introduce the notions of expansion and reduction of a closed set.

\begin{definition}[Expansion and Reduction of a Set]
    Given a compact set $q\subset\reals^{n}$ and a set of $n$ scalars $c = \{c_1,\dots,c_n\}$, where $c_i \geq 0$, the expansion of $q$ by $c$ is defined as
    \begin{equation*}
        \overline{q}(c) = \{x\in\reals^n \mid \exists x_q \in q \;\; s.t. \;\; |x_q^{(i)}-x^{(i)}|\leq c_i \;\; \forall i=\{1,\dots,n\}\},
    \end{equation*}
    and the reduction of $q$ by $c$ is 
    \begin{equation*}
        \underline{q}(c) = \{x_q \in q \mid  \forall x_{\partial q} \in \partial q, \;\; |x_q^{(i)}-x_{\partial q}^{(i)}|> c_i \;\; \forall i=\{1,\dots,n\}\},
    \end{equation*}
    where $\partial q$ denotes the boundary of $q$.
\end{definition}
In addition, we define the image of region $\qimdp$ under the learned dynamics by 
$$Im(\qimdp)=\{f_u(x)+\gest_u(x) \mid x\in\qimdp\}$$ 
and the intersection indicator function as
\begin{equation*}
    \interindicator{V}{W} = \begin{cases} 1 &\text{ if } V\cap W \neq \emptyset\\ 0 & \text{ otherwise }\end{cases}
\end{equation*}
for arbitrary sets $V$ and $W$. 
We can now bound the transition probabilities between the $\imdp$ states in $\Qimdp_\safe$. 

\begin{theorem}\label{thm:transition}
    Let $\|h\|^\qimdp_{\infty}\equiv \sup_{x\in\qimdp} |h(x)|$.
    Given an action (mode) $\aimdp\in\Aimdp$, regions $\qimdp,\qimdp' \in \Qimdp_\safe$, dataset $\dataset$, regression $\gest_u$, and 
    positive real vectors $\distBound\in\reals^n$ and $\eta\in\reals^n$, it holds that

    \begin{align}\label{eq:upper-bound}
        \begin{split}
            \max_{x\in \qimdp} &\; T^{\aimdp}(\qimdp' \mid x) \\ 
            &\leq \big(\interindicator{\overline {\qimdp}' (\distBound+\eta)}{ Im(\qimdp)} \prod_{i=1}^n P[|v^{(i)}|\leq \eta_i]+\prod_{i=1}^n P[|v^{(i)}|> \eta_i]\big)~\cdot \\ &\prod_{i=1}^n P[\|g_u^{(i)}-\gest_u^{(i)}\|^{\qimdp}_\infty \leq \distBound_i]+\prod_{i=1}^n(1-P[\|g_u^{(i)}-\gest_u^{(i)}\|^{\qimdp}_\infty \leq \distBound_i]),
        \end{split}
    \end{align}
    \begin{align}\label{eq:lower-bound}
        \begin{split}
            \min_{x\in \qimdp} &\; T^{\aimdp}(\qimdp'\mid x)\geq  \\&
            \big(1-\interindicator{\FullSet \setminus \underline {\qimdp}'(\distBound+\eta)}{Im(\qimdp)}\big) \prod_{i=1}^n P[\|g_u^{(i)}-\gest_u^{(i)}\|^{\qimdp}_\infty \leq \distBound_i]\prod_{i=1}^n P[|v^{(i)}|\leq \eta_i].
        \end{split}
    \end{align}
    \end{theorem}
    \begin{proof}[Proof] 
\noindent
\newcommand{\pshorthand}{P[\pathX(1) \in\qimdp'\mid x,\aimdp]}
\newcommand{\pshorthandgest}{P[\pathX(1) \in\qimdp'\mid x,\aimdp, \|g_u-\gest_u\| \leq \distBound]}
\newcommand{\pshorthandgestcomp}{P[\pathX(1) \in\qimdp'\mid x,\aimdp, \|g_u-\gest_u\| > \distBound]}
Let $\|g_u-\gest_u\| \leq \distBound$ denote the event $\|g_u^{(i)}-\gest_u^{(i)}\| \leq \distBound_i$ for $i=1,\dots,n$ (and similar for the complementary event). Define 
$$\pshorthand \coloneqq P[\pathX(1) \in \qimdp' | \pathX(0) = x \in \qimdp, \aimdp].$$ 
Then using the law of total probability
\begin{align*}
    &\pshorthand =\\
    &\pshorthandgest \prod_{i=1}^n P[\|g_u^{(i)}-\gest_u^{(i)}\|^{\qimdp}_\infty \leq  \distBound_i] + \\
    &\pshorthandgestcomp\prod_{i=1}^n P[\|g_u^{(i)}-\gest_u^{(i)}\|^{\qimdp}_\infty > \distBound_i]
\end{align*}
The transition kernel can be upper bounded by
\begin{align*}
    &\max_{x\in \qimdp} T^{\aimdp}(\qimdp' \mid x) = \max_{x\in \qimdp} \pshorthand \leq\\
    &\max_{x\in \qimdp} \pshorthandgest\prod_{i=1}^n P[\|g_u^{(i)}-\gest_u^{(i)}\|^{\qimdp}_\infty \leq  \distBound_i] \\
    & \hspace{35mm} + 1\cdot  \prod_{i=1}^n(1-P[\|g_u^{(i)}-\gest_u^{(i)}\|^{\qimdp}_\infty \leq  \distBound_i]).
\end{align*}
To account for the uncertainty in the regression and process noise, we again use the law of total probability and expand $\qimdp'$ by $\distBound$ and $\eta$ and check for an intersection between $Im(\qimdp)$ and $\overline {\qimdp}' (\distBound+\eta)$:
\begin{align*}
    &\leq \big(\interindicator{\overline {\qimdp}' (\distBound+\eta)}{ Im(\qimdp)} \prod_{i=1}^n P[|v^{(i)}|\leq \eta_i] + 1\cdot\prod_{i=1}^n P[|v^{(i)}|> \eta_i]\big)~\cdot \\  & \prod_{i=1}^nP[\|g_u^{(i)}-\gest_u^{(i)}\|^{\qimdp}_\infty \leq  \distBound_i] +
     \prod_{i=1}^n(1- P[\|g_u^{(i)}-\gest_u^{(i)}\|^{\qimdp}_\infty \leq  \distBound_i]).
\end{align*}
Similarly, the transition kernel can be lower bounded by determining if any points lie outside of the intersection of $Im(\qimdp)$ and $\underline {\qimdp}'(\distBound+\eta)$:
\begin{align*}
&\min_{x\in \qimdp} T^{\aimdp}(\qimdp'\mid x) = \min_{x\in \qimdp} \pshorthand\\
    &\geq\Big(\big(1-\interindicator{\FullSet \setminus \underline {\qimdp}(\distBound+\eta)}{Im(\qimdp)}\big) \prod_{i=1}^n P([v^{(i)}|\leq \eta_i] + 0\cdot\prod_{i=1}^n P[|v^{(i)}|> \eta_i]\Big)~\cdot \\  &
    \prod_{i=1}^n P[\|g_u^{(i)}-\gest_u^{(i)}\|^{\qimdp}_\infty \leq  \distBound_i] +  0\cdot \prod_{i=1}^n P[\|g_u^{(i)}-\gest_u^{(i)}\|^{\qimdp}_\infty >  \distBound_i]\\ 
    &=\big(1-\interindicator{\FullSet \setminus \underline {\qimdp}'(\distBound+\eta)}{Im(\qimdp)}\big) \prod_{i=1}^n P[|v^{(i)}|\leq \eta_i]  \prod_{i=1}^n P[\|g_u^{(i)}-\gest_u^{(i)}\|^{\qimdp}_\infty \leq  \distBound_i].
\end{align*}
\end{proof}
Theorem \ref{thm:transition} computes formal bounds for the transition probabilities by using the law of total probability with respect to the events $|v^{(i)}|\leq \eta_i$ (noise is bounded by $\eta_i$), $\|g_u^{(i)}-\gest_u^{(i)}\|^{\qimdp}_\infty \leq \distBound_i$ (the supremum of the regression error is bounded by $\distBound_i$ for all $x\in \qimdp$), and their complementary events. 
In particular, a bound on the probability of the latter event can be obtained by Lemma \ref{th:RKHS}, while the probability of former depends on the known distribution of the noise $\distribution_{v}$. 

In order to get non-trivial transition bounds, constants $\eta$ and $\distBound$ should be selected to minimize or maximize  the bounds in \eqref{eq:upper-bound} and \eqref{eq:lower-bound}  respectively. In particular, we pick $\eta$ as the smallest constants such that the noise is bounded by $\eta$ with high probability, e.g., $0.99$. Then, for this $\eta$, 
our procedure to choose a value for $\distBound$ is as follows.
We first check if $Im(\qimdp) \subset \qimdp'$. If it is the case, we pick $\epsilon$ as the greatest constants such  that $Im(\qimdp)\subset \underline {\qimdp}'(\distBound+\eta)$. Otherwise, we simply select $\epsilon$ as the smallest constants such that $\|g_u^{(i)}-\gest_u^{(i)}\|^{\qimdp}_\infty\leq \distBound$ with high probability, e.g., that satisfies the bound in Lemma \ref{th:RKHS} with probability $0.99$. 

Note that for $\qimdp\subset \FullSet$ the above procedure requires one to compute $Im(\qimdp)$. This is equivalent to computing the maximum and minimum of \eqref{eq:post-mean} for all $x\in \qimdp$. Arbitrarily tight bounds for these quantities can be computed by utilizing the convexity of most used kernels, such as the the squared-exponential function, as outlined in~\cite{cardelli2019robustness,blaas2020adversarial}. With a similar approach, a bound for $\max_{x\in\qimdp}\sigma_D(x)$ can also be computed, as this is required for the computation of Lemma \ref{th:RKHS}.

\subsubsection{Transitions to $\qunsafe$}

The probability interval for transitioning to the state $\qunsafe \in \Qimdp$, i.e., the region outside of $\FullSet$, is given by
\begin{align*}
    \Plow(\qimdp, \aimdp, \qunsafe) = 1 - \max_{x\in \qimdp} T^{\aimdp}(\FullSet \mid x),\\
    \Pup(\qimdp, \aimdp, \qunsafe) = 1 - \min_{x\in \qimdp} T^{\aimdp}(\FullSet \mid x).
\end{align*}
These bounds can be calculated as a corollary of Theorem \ref{thm:transition}. 

\begin{corollary}
Let $q \in Q_\safe$, then for any $\distBound, \eta>0$ it holds that
\begin{align*}
    &\Plow(\qimdp,\aimdp,\qunsafe) \\
    &\geq 1- \big (\interindicator{\overline {\FullSet} (\distBound+\eta)}{ Im(\qimdp)} \prod_{i=1}^n P[|v^{(i)}|\leq \eta_i] + \prod_{i=1}^n P[|v^{(i)}|> \eta_i]\big)~\cdot\\ &\prod_{i=1}^n P[\|g_u^{(i)}-\gest_u^{(i)}\|^{\qimdp}_\infty \leq \distBound_i]-\prod_{i=1}^n(1-P[\|g_u^{(i)}-\gest_u^{(i)}\|^{\qimdp}_\infty \leq \distBound_i]),\\
    &\Pup(\qimdp,\aimdp,\qunsafe) \\
    &\leq 1-\big(1-\interindicator{\FullSet \setminus \underline {\FullSet}(\distBound+\eta)}{Im(\qimdp)}\big ) \prod_{i=1}^n P[\|g_u^{(i)}-\gest_u^{(i)}\|^{\qimdp}_\infty \leq \distBound_i]\prod_{i=1}^n P[|v^{(i)}|\leq \eta_i].
\end{align*}
\end{corollary}

To complete the construction of abstraction $\I$, we make $\qunsafe$ absorbing, i.e., $\Plow(\qunsafe,\aimdp,\qunsafe)= \Pup(\qunsafe,\aimdp,\qunsafe) = 1$ for all $\aimdp \in A$, to ensure that  $\I$ does not account for the transitions to $\FullSet$ from $\qunsafe$ since such paths do not satisfy $\Spec$ as defined in \eqref{eq:satisfying-path}.

\section{Strategy Synthesis}
\label{sec:synthesis}

Given an \LTLf formula $\Spec$, ideally we would like to synthesize an optimal switching strategy $\optpolicy_{\x}$ for Process \eqref{eq:system}, under which the probability of satisfaction of $\Spec$ by the paths of Process \eqref{eq:system} is maximized. 
However, since $g_u$ is unknown, this is generally infeasible. Instead, we employ the \imdp abstraction $\I$ as constructed above, which is a conservative model of Process \eqref{eq:system} since the transition probabilities of $\I$ include uncertainties (errors) of the learning process as well as those related to the discretization. On this model, we find a strategy that is robust against all these uncertainties and maximizes the probability of satisfying $\Spec$.  Then, we can refine this strategy to a switching strategy for Process \eqref{eq:system}. Note that the resulting strategy is not necessarily optimal for Process \eqref{eq:system}, however, in what follows we  show how the error between the resulting strategy and optimal strategy $\optpolicy_{\x}$ can be quantified. 

\subsection{Near-optimal Robust Strategy}
\label{sec:robust-near-optimal-strategy}
The uncertainties in $\I$ are characterized by adversary $\adversary$, which chooses a feasible transition probability from one \imdp state to another under a given action.  
Recall that given a strategy $\policy$ and an adversary $\adversary$, $\I$ becomes a Markov chain with a well-defined probability measure over its paths.  Then, our (robust and near-optimal strategy) objective translates to finding a strategy that maximizes the probability of satisfying $\Spec$ with the assumption that the adversary (uncertainty) attempts to minimize this probability, i.e.,
\begin{equation}
    \label{eq:strategy-near-optimal}
    \nearoptpolicy = \arg\max_{\policy \in \Setofpolicies} \;\, \min_{\adversary \in \Adversary} P[\pathimdp \models \Spec \mid \policy, \adversary, \pathimdp(0) = \qimdp],
\end{equation}
Under $\nearoptpolicy$, the lower bound and upper bound on the probability of satisfaction are then given by
\begin{align}
    \label{eq:plow}
    \plow(\qimdp) &= \min_{\adversary \in \Adversary} P[\pathimdp \models \Spec \mid \nearoptpolicy, \adversary, \pathimdp(0) = \qimdp], \\
    \pup(\qimdp) &= \max_{\adversary \in \Adversary} P[\pathimdp \models \Spec \mid \nearoptpolicy, \adversary, \pathimdp(0) = \qimdp],
    \label{eq:pup}
\end{align}
respectively.  

To correctly refine a strategy computed on $\I$ to a switching strategy for process $\x$, let $z: \reals^n \to Q$ be a function that maps each state $x$ of Process \eqref{eq:system} to its corresponding discrete region $q \in Q$, i.e., $z(x) = q$ iff $x \in q$.  We also use $z$ to denote mapping from finite paths of process $\x$ to finite paths of $\I$, i.e., for a finite path $\pathX^{k}= x_0 \xrightarrow{u_0} x_1 \xrightarrow{u_1}  \ldots \xrightarrow{u_{k-1}} x_k$,  the corresponding path on $\I$ is given by
$$z(\pathX^k) = z(x_0) \xrightarrow{u_0} z(x_1) \xrightarrow{u_1}  \ldots \xrightarrow{u_{k-1}} z(x_k).$$
Then, strategy $\nearoptpolicy$ on $\I$ is correctly refined to a switching strategy $\nearoptpolicy_\x$ for process $\x$ by
\begin{equation}
    \label{eq:mapping-strategy}
    \nearoptpolicy_\x(\pathX^k) = \nearoptpolicy(z(\pathX^k)).
\end{equation}

Note that, the maximum probability of satisfaction of $\Spec$ by Process \eqref{eq:system} is necessarily lower bounded by $\plow$ in \eqref{eq:lower-bound}, i.e., $\popt(x) \geq \plow(z(x))$, where
$$\popt(x) = \max_{\policy_\x} P[\pathX \models \Spec \mid \policy_\x, \pathX(0) = x].$$
However, $\popt(x)$ is not necessarily upper bounded by $\pup$ in \eqref{eq:pup}.  Probability $\popt(x)$ can instead be upper bounded by
\begin{equation}
    \label{eq:prob-upper-bound}
    \pup^*(x) = \max_{\policy \in \Setofpolicies} \;\, \max_{\adversary \in \Adversary} P[\pathimdp \models \Spec \mid \policy, \adversary, \pathimdp(0) = z(x)],
\end{equation}
where the adversary (uncertainty) cooperatively chooses feasible transition probabilities to maximize the probability of satisfaction of $\Spec$.  Therefore, 
$$\popt(x) \in [\plow(z(x)), \pup^*(z(x))].$$

Below, we show how the strategy in \eqref{eq:strategy-near-optimal}, its corresponding probability bounds in \eqref{eq:plow} and \eqref{eq:pup}, and probability in \eqref{eq:prob-upper-bound} can be computed.

\subsection{Synthesis}
Given an \LTLf formula $\Spec$, a deterministic finite automaton can be constructed that precisely accepts the language of $\Spec$ per \cite{de2013linear}.
\begin{definition}[\dfa]
    A \textit{deterministic finite automaton} (\dfa) constructed from an \LTLf formula $\Spec$ defined over atomic propositions $\APs$ is a tuple $\DFAA=(\DFAStates, 2^{\APs}, \DFATransition, \DFAState_0,\DFAFinalState)$, where 
    $\DFAStates$ is a finite set of states,
    $2^{\APs}$ is a finite set of input symbols, each of which is a set of atomic propositions in $\APs$,
    $\DFATransition:\DFAStates\times 2^{\APs} \to\DFAStates$ is the transition function, 
    $\DFAState_0 \in \DFAStates$ is the initial state, and 
    $\DFAFinalState \subseteq \DFAStates$ is the set of accepting (final) states. 
\end{definition} 
A finite \emph{run} on a \dfa is a sequence of states $s = \DFAState_0\DFAState_1\dots\DFAState_{n+1}$ induced by a sequence of symbols $\trace = \trace_0\trace_1\dots\trace_n$, where $\trace_i \in 2^{\APs}$ and $\DFAState_{i+1} = \DFATransition(s_i,\trace_i)$. Finite run $s$ on trace $\trace$ is accepting if $\DFAState_{n}\in\DFAFinalState$.  If $s$ is accepting, then trace $\trace$ is accepted by $\DFAA$.  The set of all traces that are accepted by $\DFAA$ is call the language of $\DFAA$, denoted by $\mathcal{L}(\DFAA)$.  This language is equal to the language of $\Spec$, i.e., $\mathcal{L}(\Spec) = \mathcal{L}(\DFAA)$. 

Next, we construct a product of \dfa $\DFAA$ with \imdp $\I$ to capture the paths of $\I$ that satisfy $\Spec$.

\begin{definition}[Product \imdp]
    Given an \imdp $\I = (\Qimdp, \Aimdp, \Plow^{\PIMDPS},$ $\Pup^{\PIMDPS}, \APs, L)$ and \dfa $\DFAA = (\DFAStates, 2^{\APs}, \DFATransition, \DFAState_0,\DFAFinalState)$, the \textit{product \imdp} (\pimdp) $\PIMDPS = \I \times \DFAA$ is itself an \imdp defined as the tuple $\PIMDPS=(Q^\PIMDPS, \Aimdp, \Plow^{\PIMDPS}, \Pup^{\PIMDPS}, Q_0^\PIMDPS, Q_F^\PIMDPS)$, where 
    $Q^\PIMDPS = \Qimdp\times\DFAStates$, $Q_F^\PIMDPS = \Qimdp\times\DFAStates_F,$
    $$Q_0^\PIMDPS = \set{(\qimdp,\DFAState_{init}) \mid  \qimdp \in \Qimdp , \; \DFAState_{init}=\DFATransition(\DFAState_0, L(\qimdp))},$$
    and
    \begin{align*}
        \Plow^{\PIMDPS}((\qimdp, \DFAState), \aimdp, (\qimdp', \DFAState')) &= \begin{cases}
            \Plow(\qimdp,\aimdp,\qimdp') & \text { if } \DFAState'=\DFATransition(\DFAState,L(\qimdp))\\
            0 & \text{ otherwise}
        \end{cases} \\
        \Pup^{\PIMDPS}((\qimdp, \DFAState), \aimdp, (\qimdp', \DFAState')) &= \begin{cases}
            \Pup(\qimdp,\aimdp,\qimdp') & \text { if } \DFAState'=\DFATransition(\DFAState,L(\qimdp))\\
            0 & \text{ otherwise}.
        \end{cases} 
    \end{align*}
\end{definition}

\noindent
In its essence, the \pimdp incorporates the historical dependencies on the runs of the \dfa and couples them with the paths of the \imdp. 
The projection of a path of $\PIMDPS$ that reaches a state in $Q^\PIMDPS_F$ onto $\DFAA$ is an accepting run, and hence, the projection of this path onto abstraction $\I$ is a path that satisfies $\Spec$.
Therefore, the synthesis problem in \eqref{eq:strategy-near-optimal} is reduced to computing a robust (pessimistic) strategy on product $\PIMDPS$ that maximizes the probability of reaching $Q^\PIMDPS_F$.  Similarly, the probability in \eqref{eq:prob-upper-bound} is given by an optimistic strategy on $\PIMDPS$ that maximizes the the probability of reaching $Q^\PIMDPS_F$.  These problems are variations of a known problem called \textit{maximal reachability probability problem} and can be solved using a method similar to value iteration called interval-value iteration \cite{wu2008reachability,lahijanian2015formal}, whose computational complexity is polynomial.  
The resulting strategies are stationary on $\PIMDPS$, which map to history dependent strategies on $\I$.

\subsection{Correctness}
The following theorem shows that $\nearoptpolicy_\x$ in \eqref{eq:mapping-strategy} is a $\opterror$-near-optimal switching strategy for Process~\eqref{eq:system} and quantifies its distance (error) $\opterror$ to the optimal switching strategy $\optpolicy_\x$.

\begin{theorem}
    \label{thm:correctness}
    Given a partially-known switched stochastic system as defined in Process \eqref{eq:system}, a dataset $D$,
    a compact set $\FullSet \subset \reals^n$, an \LTLf formula $\Spec$ defined over the regions of interest in $\FullSet$, let $\I$ be an \imdp abstraction as described in Section~\ref{sec:IMDP}, $\nearoptpolicy$ be a strategy on $\I$ given by \eqref{eq:strategy-near-optimal}, and $\nearoptpolicy_\x$ be the switching strategy for Process \eqref{eq:system} obtained from $\nearoptpolicy$ according to mapping $z$ in \eqref{eq:mapping-strategy}.  Further, let $\plow$, $\pup$, and $\pup^*$ be the probability bounds in \eqref{eq:plow}, \eqref{eq:pup}, and \eqref{eq:prob-upper-bound}, respectively.  Then, it holds that
    \begin{equation*}
        P[\pathX \models \Spec \mid \nearoptpolicy_\x, \pathX(0)=x] \in [\plow(z(x)), \pup(z(x))],
    \end{equation*}
    and 
    \begin{equation*}
        | P[\pathX \models \Spec \mid \nearoptpolicy_\x, \pathX(0)=x] - \popt(x) | \leq \pup^*(z(x)) - \plow(z(x)),
    \end{equation*}
    where $\popt(x) = \max_{\policy_\x \in \Setofpolicies_\x} P[\pathX \models \Spec \mid \policy_\x, \pathX(0)=x]$.
\end{theorem}

\noindent
    Theorem \ref{thm:correctness} is a straightforward consequence of Theorem \ref{thm:transition} and guarantees that the probability that Process \eqref{eq:system} satisfies $\Spec$  is contained between $\plow$ and $\pup$. In order to quantify the distance of $\pi^{\epsilon}_{\x}$ from the optimal strategy $\optpolicy_\x$, we compute the optimal upper bound probability $\pup^*$ correspondent to the strategy that optimistically maximizes the probability of reaching $Q^\PIMDPS_F$. In fact, recall that $\nearoptpolicy_{\x}$ corresponds to the strategy that maximizes the lower bound of reaching $Q^\PIMDPS_F.$ It follows  that  for any $x \in \FullSet,$  
    $$ \opterror_{x} = |\pup^*(z(x)) - \plow(z(x))|.$$

Given a probability bound $\pthresh$ on the satisfaction of formula $\Spec$, we use $\plow$ and $\pup$ to classify each initial state $x_0 \in \FullSet$ as one of the following:
\begin{equation*}
    x \in 
    \begin{cases}
      \Qyes & \text{ if }\plow(z(x)) \geq \pthresh \\
      \Qno & \text{ if }\pup(z(x)) < \pthresh \\
      \Qposs & \text{ otherwise. }
    \end{cases}
\end{equation*}
Given initial state $x_0$, we can guarantee that $\Spec$ is definitely satisfied by the underlying system with at least $\pthresh$ if $x_0 \in \Qyes$.  If $x_0\in \Qno$, then we can guarantee that the underlying system never meets the probability threshold $\pthresh$.
For $x_0 \in \Qposs$, no guarantees relative to the threshold $\pthresh$ can be given. 

%% file: sections/5_examples.tex
\section{Case Studies}
\label{sec:case-studies}
We illustrate the proposed framework in three case studies using linear and nonlinear switched systems.
In all the demonstrations, the compact set is $\FullSet=[-2,2]\times[-2,2]$. 
We use a uniform grid of size 0.125 over $\FullSet$ to create $\Qimdp_\safe$ for our abstraction. 

\subsection{Linear Switched System with Three Modes }
We first demonstrate the framework on a three-mode linear switched system similar to the synthesis example presented in \cite{lahijanian2015formal}. We assume the dynamics in all three modes are unknown, i.e.,
\begin{equation*}
    \x_{k+1} = g_{\pu_k}(\x_k) + \noise_k,
\end{equation*}
where each mode is a linear system with $g_u(\x_k) = A_u \x_k$ for all $u \in \{1,2,3\}$, 
\begin{equation*}
    A_1 = \begin{bmatrix}
    0.4 & 0.1\\
    0 & 0.5
    \end{bmatrix},\;\;
    A_2 = \begin{bmatrix}
    0.4 & 0.5\\
    0 & 0.5
    \end{bmatrix},\;\;
    A_3 = \begin{bmatrix}
    0.4 & 0\\
    0.5 & 0.5
    \end{bmatrix},
\end{equation*}
and $\noise$ is drawn from a Gaussian distribution $\mathcal{N}(0,\sigma^2 I)$ truncated between $[-\sigma, \sigma]$ with $\sigma = 0.01$. 

Two-hundred i.i.d. data points per mode were sampled and propagated through the dynamics to create the dataset for regression.  
Figure~\ref{fig:three-mode} shows the partition of $\FullSet$ with labelled regions $Des$ and $Obs$ indicating ``Desired'' and ``Obstacle'' regions, respectively. 
With an abuse of notation, these are used to define the atomic propositions $\{Des, Obs\}$ and the \LTLf specification $$\Spec_1= \Globally (\neg Obs) \wedge \Eventually(Des),$$ 
which reads, ``Globally avoid Obstacles and eventually reach a Desired region''.

Using our framework, we learned the unknown dynamics and synthesized a robust and near-optimal switching strategy $\nearoptpolicy$.
Figure~\ref{fig:three-mode} shows the classification of each initial region with threshold probability $\pthresh=0.95$ under this strategy. Initial states with the $Des$ label belong to $\Qyes$ as they satisfy $\Spec_1$ while states with the $Obs$ label violate it and belong to $\Qno$. There are additional states belonging to $\Qyes$ such that actions dictated by $\nearoptpolicy$ drive the system to an accepting state with a high probability. These results closely resemble the results presented in \cite{lahijanian2015formal}, which assumed full knowledge of the dynamics, whereas here the dynamics are fully unknown and are estimated from a limited set of data.

\begin{figure}
    
\end{figure}

\begin{figure}
    \centering
    \newcommand\figwidth{0.23\textwidth}
    \newcommand\scaleoffset{\hspace*{22pt}\vspace*{-2pt}}
    \begin{subfigure}[b]{0.20\textwidth}
    \includegraphics[width=\textwidth]{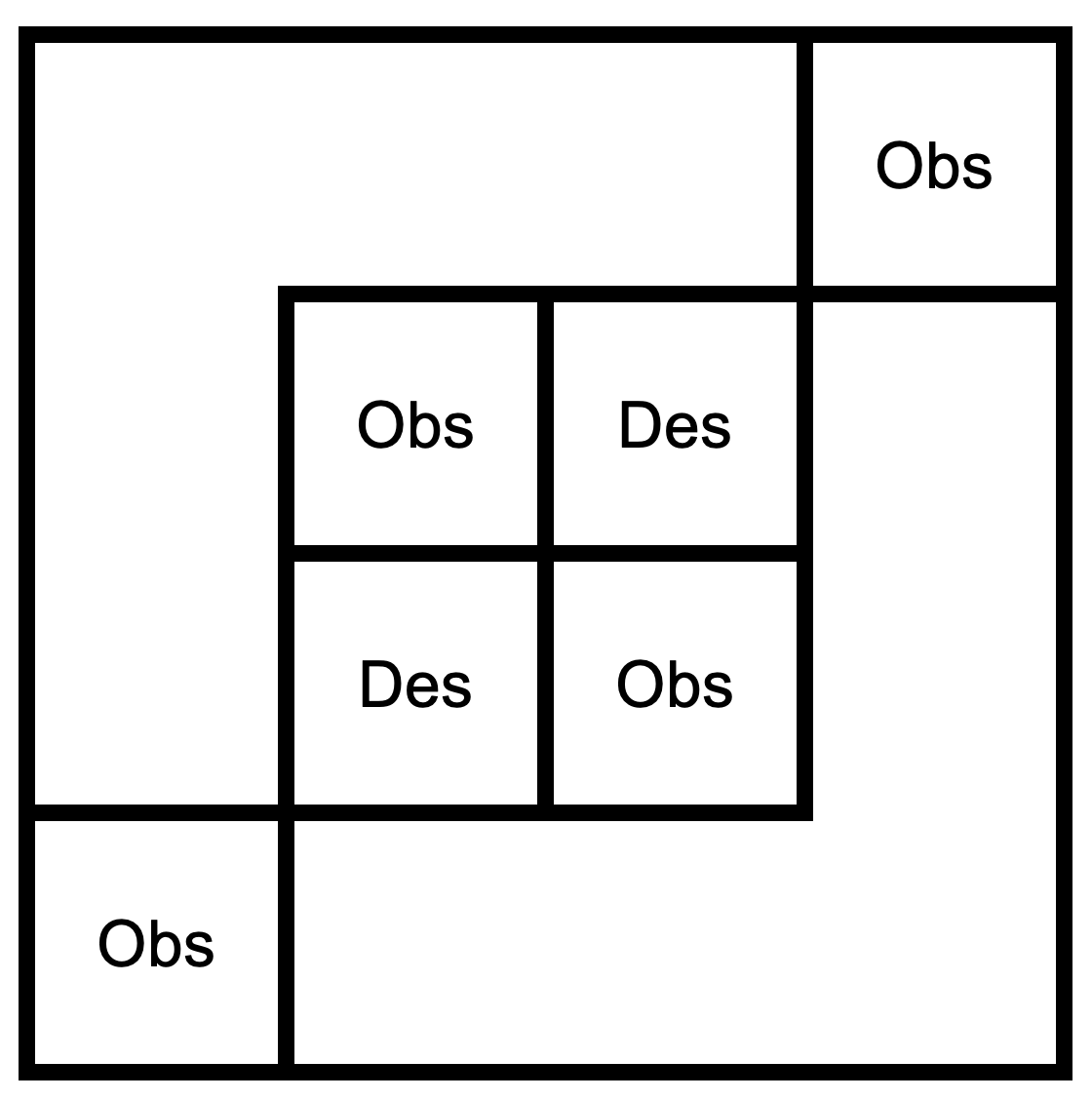}
    \label{fig:partition}
    \end{subfigure}%
    \begin{subfigure}[b]{\figwidth}
    \begin{subfigure}[b]{\textwidth}
    \scaleoffset\includegraphics[width=0.7\linewidth]{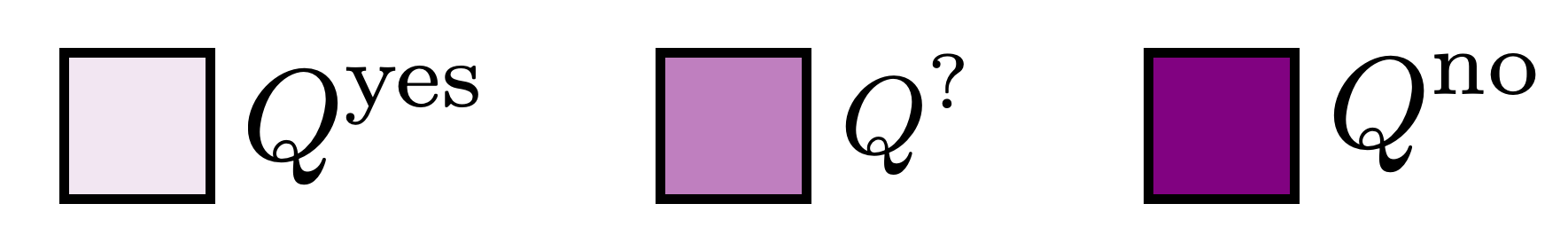}
    \end{subfigure}
    \includegraphics[width=\textwidth]{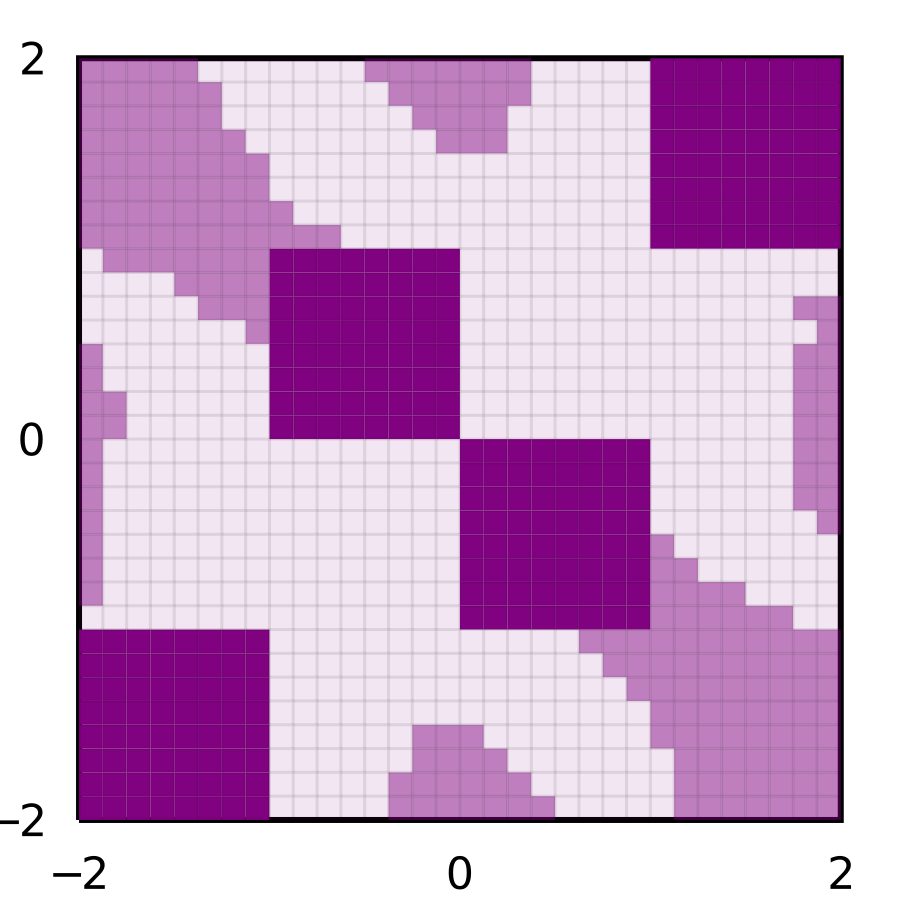}
    \end{subfigure}
    \caption{Region partition and classification of initial states using the strategy synthesized for the linear system and $\Spec_1$.}
    \label{fig:three-mode}
\end{figure}

\subsection{Parameter Choices}\label{sec:noise-demo}
We provide a brief look at the effect of choosing different values of $\eta$, the bounds on the noise components in Theorem~\ref{thm:transition}, on the synthesis results. For any choice of $\eta$, 
the optimal value of $\distBound$, the bounds on the suprema of the regression error components,
is then chosen to minimize (maximize) the upper-bound (lower-bound) of the transition probability in Theorem~\ref{thm:transition} as discussed previously. In the ideal case, the noise parameter primarily effects the lower-bound of the transition probability, as the optimal choice of $\distBound$ leaves the indicator function in \eqref{eq:lower-bound} with a value of zero. 

For the three-mode linear switched system above, the effect of changing $\eta$ uniformly in a na\"ive manner is shown in Figure~\ref{fig:noise-param}. The choices of $\eta$ are presented as fractions of the bounds on the truncated Gaussian distribution. As expected, the lower-bound on the probability of satisfaction decreases as the value of $\eta$ decreases. In a large part, this is due to the reduction in the lower-bound of the probability of staying within $\FullSet$. There is clearly an optimal trade-off to be made between $\distBound$ and $\eta$, which will be considered in future work.

\begin{figure}
    \newcommand\figwidth{0.23\textwidth}
    \includegraphics[width=\figwidth]{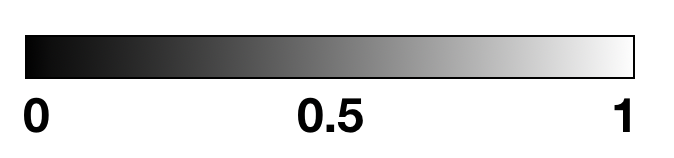}
    \begin{subfigure}{\figwidth}
    \includegraphics[width=\linewidth]{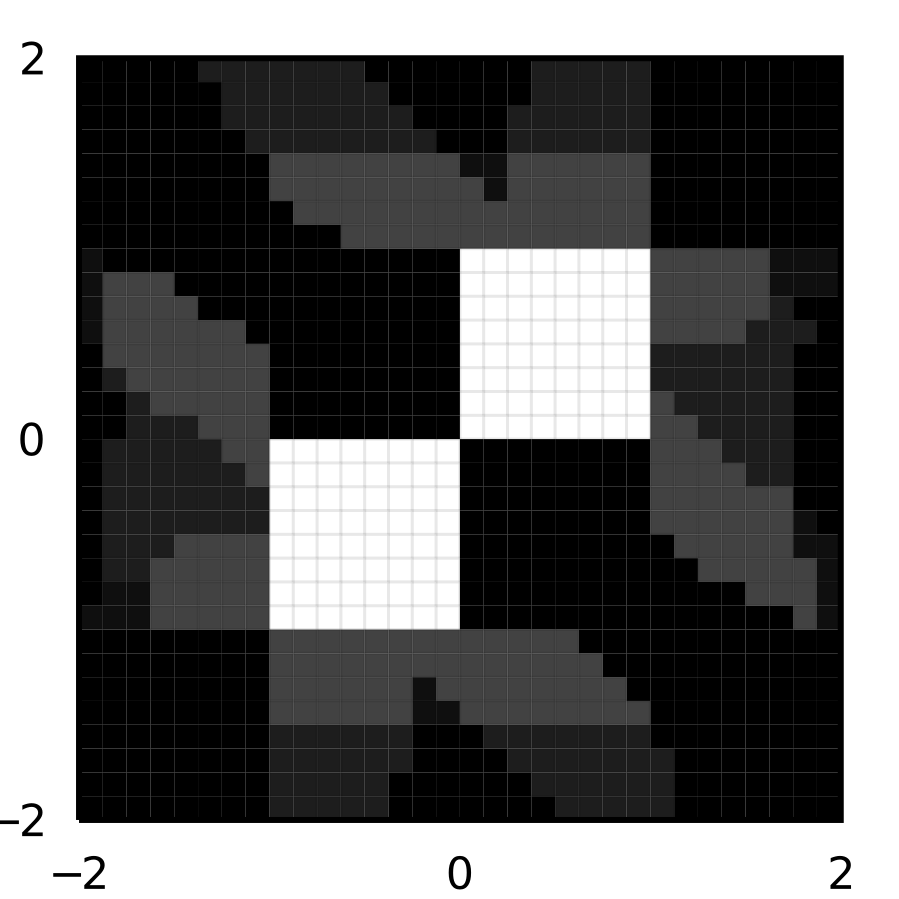}
    \caption{$\eta=0.96\sigma$}
    \end{subfigure}%
    \begin{subfigure}{\figwidth}
    \includegraphics[width=\linewidth]{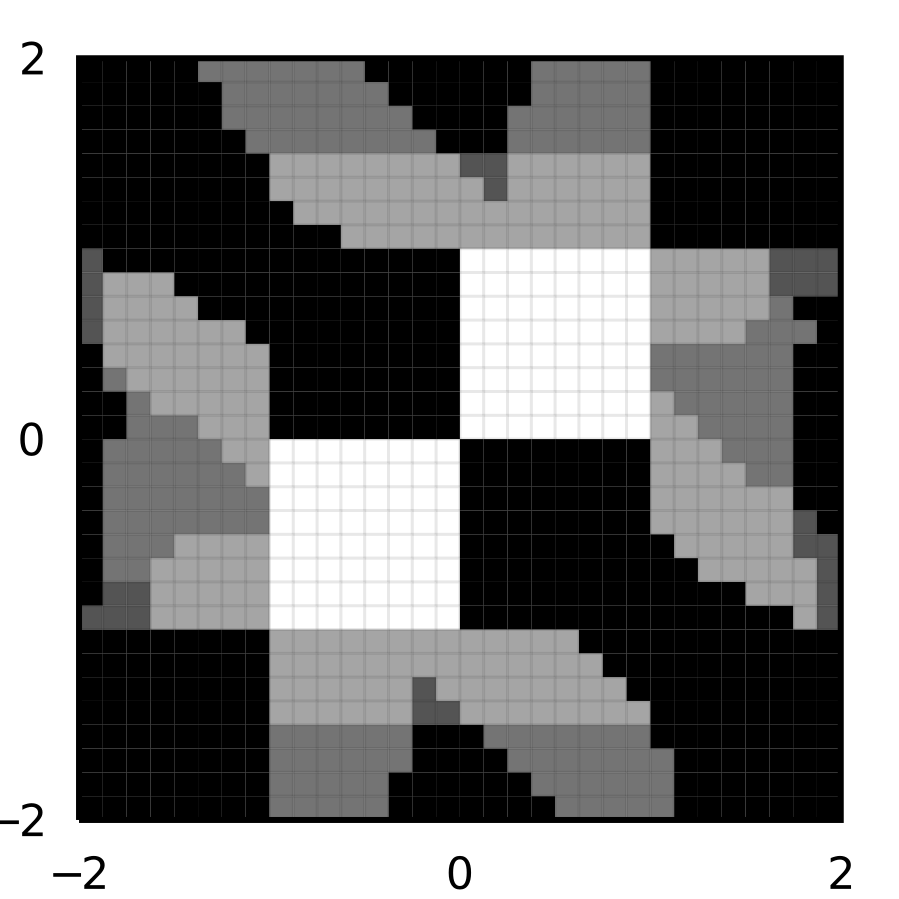}
    \caption{$\eta=0.97\sigma$}
    \end{subfigure}
    \begin{subfigure}{\figwidth}
    \includegraphics[width=\linewidth]{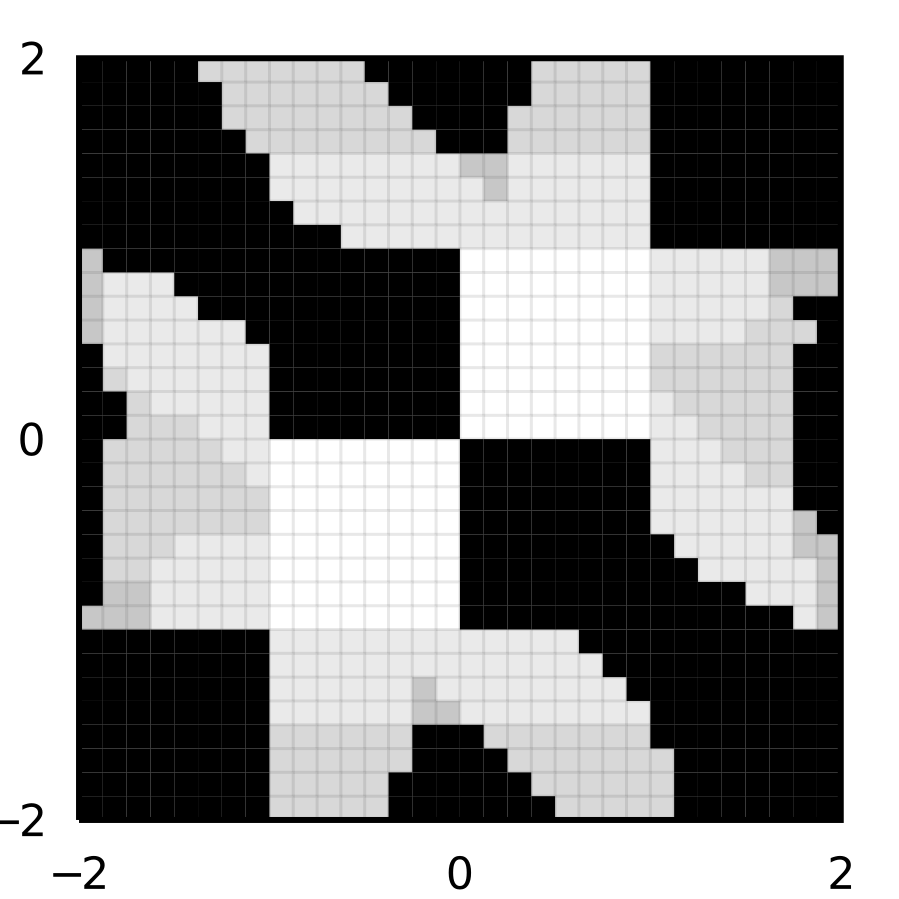}
    \caption{$\eta=0.985\sigma$}
    \end{subfigure}%
    \begin{subfigure}{\figwidth}
    \includegraphics[width=\linewidth]{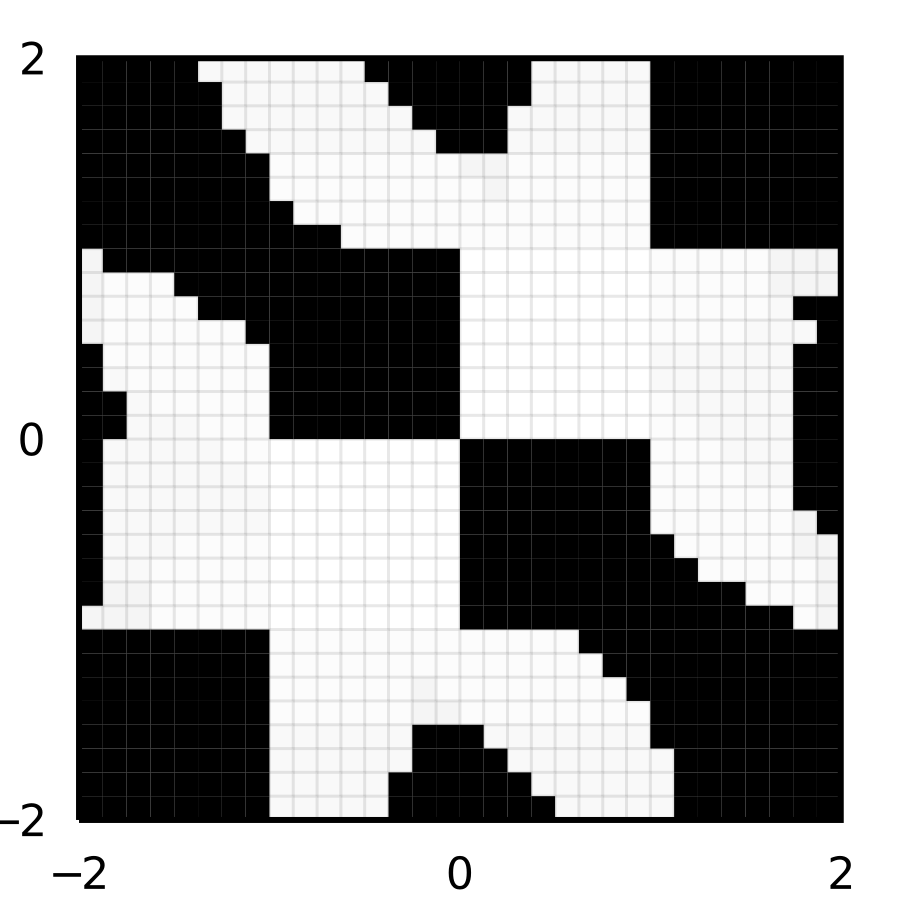}
    \caption{$\eta=0.995\sigma$}
    \end{subfigure}
    \caption{Effect on changing the parameter $\eta$ on the lower bound of the probability of satisfaction from each state.}
    \label{fig:noise-param}
\end{figure}

\subsection{Nonlinear Switched System with Four Modes}\label{sec:non-linear-ex}
Next, we synthesize a switching strategy for a nonlinear system with four modes, a known linear dynamics component, and an unknown nonlinear dynamics component. The form of the system is
\begin{equation*}
    \x_{k+1} = \x_k + g_{\pu_k}(\x_k) + \noise_k
\end{equation*}
The unknown dynamics are
\begin{equation*}
g_{u}(\x_k) =
\begin{cases}
[0.5 + 0.2\sin{\x_k^{(2)}}, 0.4\cos{\x_k^{(1)}}]^T & \text{ if } u=1\\
[-0.5 + 0.2\sin{\x_k^{(2)}}, 0.4\cos{\x_k^{(1)}}]^T & \text{ if } u=2\\
[0.4\cos{\x_k^{(2)}}, 0.5+0.2\sin{\x_k^{(1)}}]^T & \text{ if } u=3\\
[0.4\cos{\x_k^{(2)}}, -0.5+0.2\sin{\x_k^{(1)}}]^T & \text{ if } u=4
\end{cases}
\end{equation*}
where $x^{(i)}$ indicates the $i$-th component of the state. Three-hundred i.i.d. data points were generated from each mode with the same noise distribution. 

The region of interest is $Des$, which indicated by the black square in Figure~\ref{fig:nonlinear-simple}, and the specification is $$\Spec_2= \Eventually \, Des.$$
Figure~\ref{fig:nonlinear-simple} shows the synthesis and simulated result.
The black lines indicate simulated paths from a set of randomly selected initial states, and the star indicates the terminal state. Individual simulations were terminated if an accepting state in the \pimdp was reached, or if the specification was violated.
The size of $\Qyes$ in Figure~\ref{fig:nonlinear-simple} shows that the strategy $\nearoptpolicy$ can drive many states into $Des$ with a high probability. However, there are many states in $\Qno$. In these states, the \emph{upper-bound} of the probability of satisfying $\Spec_2$ induced by employing $\nearoptpolicy$ does not meet $\pthresh$. In other words, there is a significant chance of violating $\Spec_2$.

Figure~\ref{fig:gamma-example} shows upper bounds of error $\opterror$ at each state. Recall that 
$\opterror$ bounds the satisfaction probability distance under optimal and near-optimal strategies in Problem~\ref{prob:Syntesis}. 
The darker regions correspond to areas with $\opterror$ approaching zero, meaning $\nearoptpolicy$ is indeed near-optimal.
The lighter regions with $\opterror$ approaching one indicate that $\nearoptpolicy$ does not necessarily choose the optimal action.
This could be mitigated by collecting more data and performing a finer abstraction, or it is possible the system does not have sufficient control authority. 

\begin{figure}

    \newcommand\subfigwidth{\columnwidth}
    \newcommand\figwidth{0.7\columnwidth}
    \newcommand\partwidth{0.6\columnwidth}
    \newcommand\scalewidth{0.5\columnwidth}
    \newcommand\scaleoffset{\hspace*{18pt}\vspace*{-4pt}}
    
    \begin{subfigure}{\subfigwidth}
        \centering
        \scaleoffset\includegraphics[width=\scalewidth]{images/legend.png}
        \includegraphics[width=\figwidth]{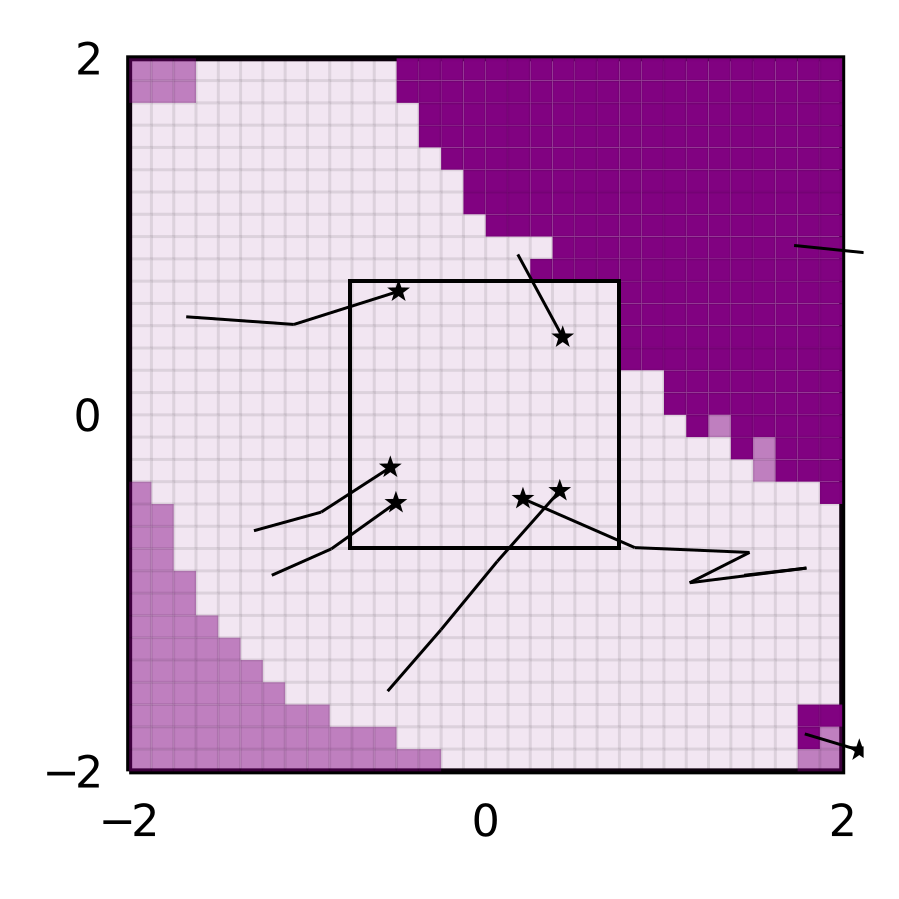}
        \caption{State classification and simulations.}
        \label{fig:nonlinear-simple}
    \end{subfigure}

    \begin{subfigure}{\subfigwidth}
        \centering
        \scaleoffset\includegraphics[width=\scalewidth]{images/scale.png}
        \includegraphics[width=\figwidth]{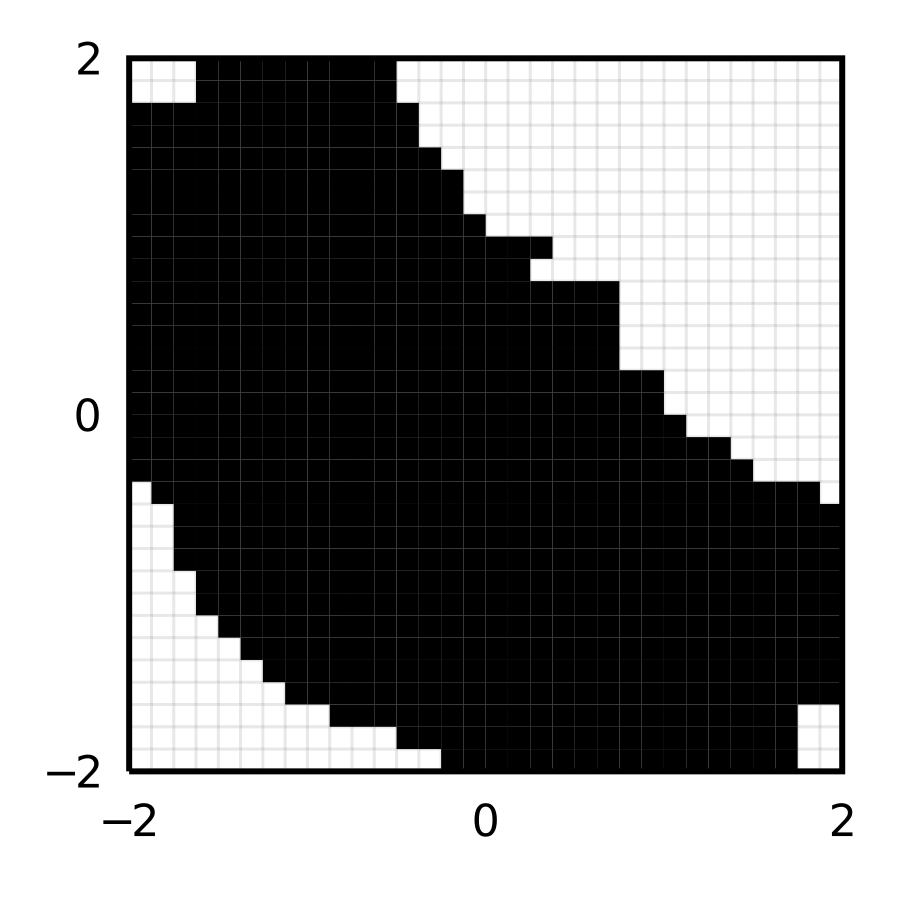}
        \caption{Optimization error $\opterror$ at each state.}
        \label{fig:gamma-example}
    \end{subfigure}
    
    \caption{Synthesis results and simulations for the nonlinear system for $\Spec_2$. The black square in (a) indicates the region with label $Des$.}
    \label{fig:nonlinear}
\end{figure}

Finally, we perform controller synthesis for the partially-known nonlinear switched system given the specification $$\Spec_3=\Globally(\neg O)\wedge \Eventually(D1)\wedge \Eventually(D2)$$ with two reachability objectives. Figure~\ref{fig:partition3} shows the partitioning of the space with labels $D1$, $D2$ and $O$ indicating ``Desired Location 1'', ``Desired Location 2'' and ``Obstacle'' regions respectively. We use the same abstraction we generated in the previous case. 

Figure~\ref{fig:nonlinear-multi-sims} shows several simulation results from different initial states, two of which are in $\Qposs$ and the others in $\Qyes$. The $\Qno$ set is comprised of only the $O$ regions, because starting in $O$ automatically violates the specification.
Much of the $\Qyes$ set is made up of a majority of $D2$ and some states starting in $D1$. There is a large amount of free space that can be driven into $D1$ with high probability. All paths but one terminate with satisfying $\Spec_3$, but a single path is driven into an obstacle.
Figure~\ref{fig:gamma-multi} shows that the optimal action has been found for many of the states, but there is a significant number of states with a trivial bound of $\opterror=1$. This metric can help identify areas for further data collection, or state discretization refinement.

\begin{figure}

    \newcommand\subfigwidth{\columnwidth}
    \newcommand\figwidth{0.7\columnwidth}
    \newcommand\partwidth{0.6\columnwidth}
    \newcommand\scalewidth{0.5\columnwidth}
    \newcommand\scaleoffset{\hspace*{18pt}\vspace*{-4pt}}
    
    \begin{subfigure}{\subfigwidth}
        \centering
        \includegraphics[width=\partwidth]{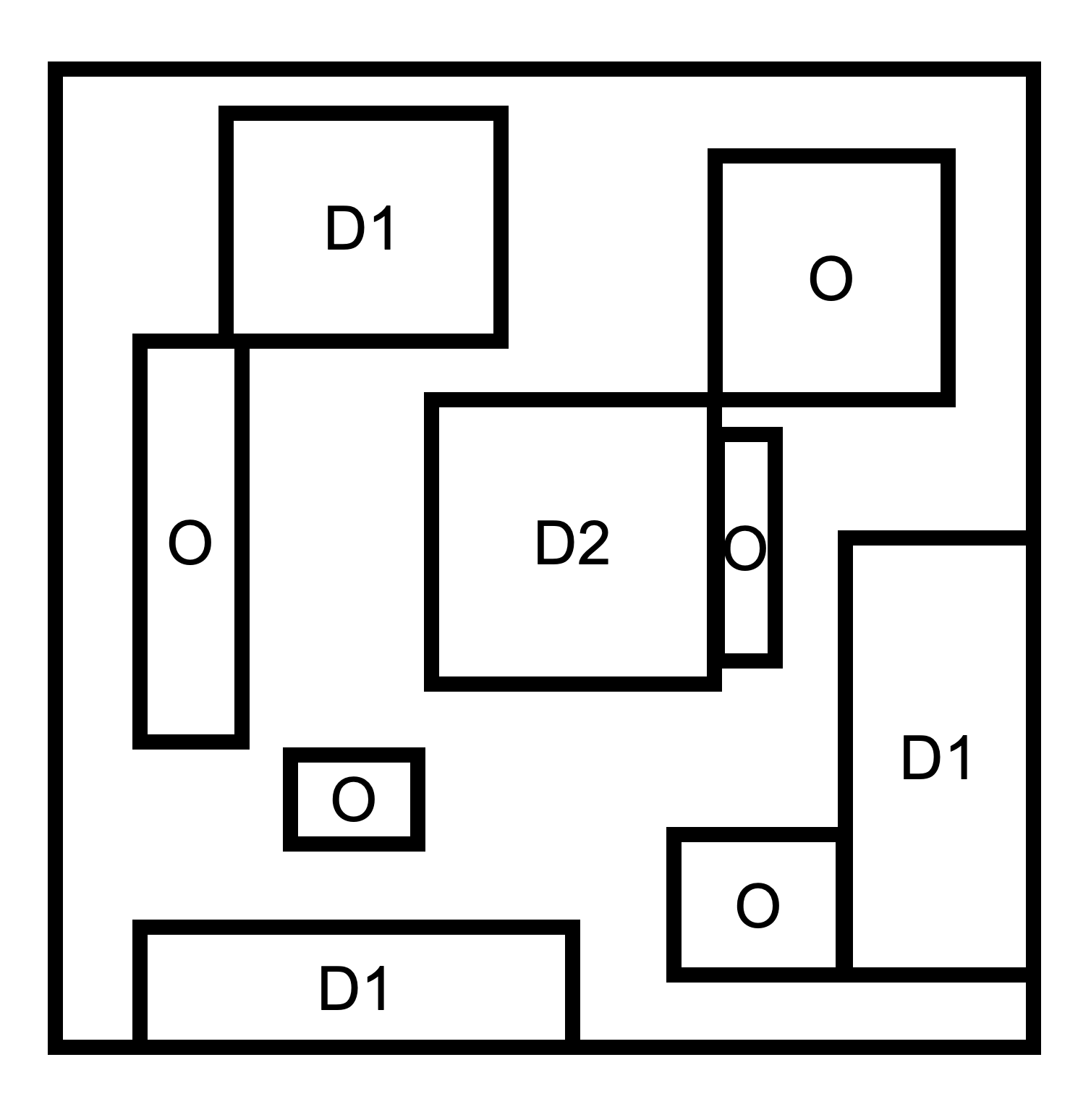}
        \caption{Regions of interest for $\Spec_3=\Globally(\neg O)\wedge \Eventually(D1)\wedge \Eventually(D2)$.}
        \label{fig:partition3}
    \end{subfigure}
    
    \begin{subfigure}{\subfigwidth}
            \centering
    \scaleoffset\includegraphics[width=\scalewidth]{images/legend.png}
             \includegraphics[width=\figwidth]{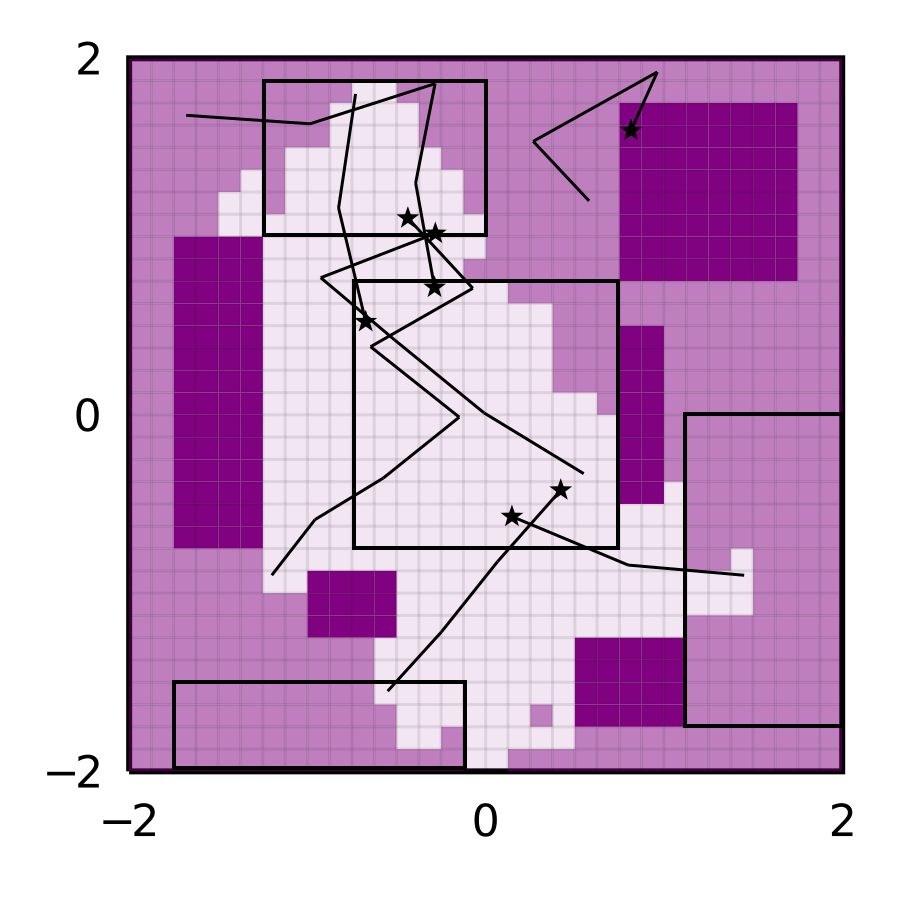}
            \caption{Classification of each state and simulated trajectories.}
            \label{fig:nonlinear-multi-sims}
    \end{subfigure}
    
    \begin{subfigure}{\subfigwidth}
        \centering
        \scaleoffset\includegraphics[width=\scalewidth]{images/scale.png}
        \includegraphics[width=\figwidth]{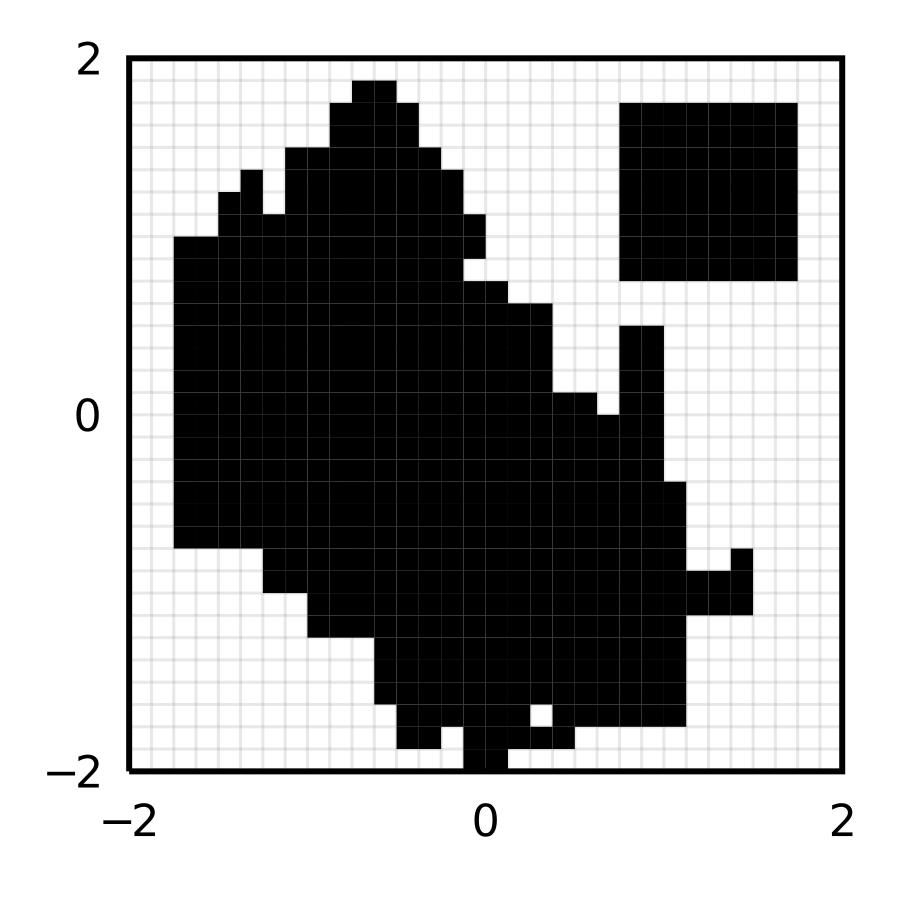}
        \caption{Optimization error $\opterror$ at each state.}
        \label{fig:gamma-multi}
    \end{subfigure}
    
    \caption{Synthesis results and simulations for the nonlinear system with $\Spec_3$.}
    \label{fig:nonlinear-multi}
\end{figure}

%% file: sections/6_conclusion.tex
\section{Conclusion}

We developed a data-driven framework for synthesizing a near-optimal control strategy for partially unknown switched stochastic  systems with \LTLf specifications. The framework is based on abstraction to an uncertain Markov model that incorporates both the uncertainty given by the stochastic dynamics of the system and the uncertainty in learning the unknown dynamics of the system via GP regression. Our work makes a step towards formally safe and correct data-driven systems. However, many challenges are ahead in order to make our framework to scale to larger datasets and higher dimensional systems. In the future, we plan to consider sparse Gaussian processes 
as well as optimal techniques for parameter tuning and refinement.